\RequirePackage{amsmath}
\documentclass{jfp1}
\usepackage{amssymb}
\usepackage{url}                  
\usepackage{listings}          
\usepackage{xspace}
\usepackage{graphicx}
\usepackage{pdfpages}
\newtheorem{thm}{\bf Theorem} 
\newtheorem{lem}[thm]{\bf Lemma}
\newtheorem{cor}{\bf Corollary}

\usepackage[inference]{semantic}

\usepackage{ifthen}
\newcommand{\isTechReport}{true} 
\newcommand\includeTechReport[1]{%
  \ifthenelse{\equal{\isTechReport}{true}}
    {{#1}}
    {\ignorespaces}
\xspace}

\title[Dynamic Witnesses for Static Type Errors]
      {Dynamic Witnesses for Static Type Errors\thanks{This work was supported by NSF grants CCF-1422471, CCF-1223850,
CCF-1218344, CCF-1116289, CCF-0954024, Air Force grant FA8750-15-2-0075,
and a gift from Microsoft Research.}\\
      (or, Ill-Typed Programs Usually Go Wrong)}

\author[E. L. Seidel, R. Jhala, and W. Weimer]
       {ERIC L. SEIDEL, RANJIT JHALA\\
         University of California, San Diego, USA\\
        WESTLEY WEIMER\\
         University of Virginia, USA
         \email{eseidel@cs.ucsd.edu, jhala@cs.ucsd.edu, weimer@virginia.edu}
       }

\usepackage{tcolorbox}

\definecolor{ocaml}{HTML}{8dd3c7}
\definecolor{sherrloc}{HTML}{FFFFB3}

\tcbset{every box/.style={on line,colframe=black,size=fbox,boxrule=1pt}}
\newtcbox{\hlOcaml}{colback=red!25,toprule=0pt,bottomrule=0pt}
\newtcbox{\hlSherrloc}{colback=red!25,leftrule=0pt,rightrule=0pt}

\usepackage{commands}
\ifdefined\withcolor
	\definecolor{haskellblue}{rgb}{0.0, 0.0, 1.0}
	\definecolor{haskellstr}{rgb}{0.2, 0.2, 0.6}
	\definecolor{haskellred}{rgb}{1.0, 0.0, 0.0}
	\definecolor{gray_ulisses}{gray}{0.55}
	\definecolor{castanho_ulisses}{rgb}{0.71,0.33,0.14}
	\definecolor{preto_ulisses}{rgb}{0.41,0.20,0.04}
	\definecolor{green_ulisses}{rgb}{0.0,0.4,0.0}
\else
	\definecolor{haskellblue}{gray}{0.1}
	\definecolor{haskellstr}{gray}{0.1}
	\definecolor{haskellred}{gray}{0.1}
	\definecolor{gray_ulisses}{gray}{0.1}
	\definecolor{castanho_ulisses}{gray}{0.1}
	\definecolor{preto_ulisses}{gray}{0.1}
	\definecolor{green_ulisses}{gray}{0.1}
\fi

\def\codesize{\normalsize}

\lstdefinelanguage{HaskellUlisses}{
	basicstyle=\ttfamily,
	sensitive=true,
	morecomment=[l][\color{gray_ulisses}\ttfamily]{--},
	morecomment=[s][\color{gray_ulisses}\ttfamily]{\{-}{-\}},
	morestring=[b]",
	stringstyle=\color{haskellstr},
	showstringspaces=false,
	numberstyle=\codesize,
	numberblanklines=true,
	showspaces=false,
	breaklines=true,
	showtabs=false,
        escapeinside={(*}{*)},%
	emph=
	{[1]
		FilePath,IOError,abs,acos,acosh,and,any,appendFile,approxRational,asTypeOf,asin,
		asinh,atan,atan2,atanh,basicIORun,break,catch,ceiling,chr,compare,concat,concatMap,
		const,cos,cosh,curry,cycle,decodeFloat,denominator,digitToInt,div,divMod,drop,
		dropWhile,either,elem,encodeFloat,enumFrom,enumFromThen,enumFromThenTo,enumFromTo,
		error,even,exp,exponent,fail,filter,flip,floatDigits,floatRadix,floatRange,floor,
		fmap,foldl,foldl1,foldr,foldr1,fromDouble,fromEnum,fromInt,fromInteger,
		fromRational,fst,gcd,getChar,getContents,getLine,head,id,inRange,index,init,intToDigit,
		interact,ioError,isAlpha,isAlphaNum,isAscii,isControl,isDenormalized,isDigit,isHexDigit,
		isIEEE,isInfinite,isLower,isNaN,isNegativeZero,isOctDigit,isPrint,isSpace,isUpper,iterate,
		last,lcm,length,lex,lexDigits,lexLitChar,lines,log,logBase,lookup,mapM,mapM_,max,
		maxBound,maximum,maybe,min,minBound,minimum,mod,negate,not,notElem,numerator,odd,
		or,pi,primExitWith,print,product,properFraction,putChar,putStr,putStrLn,quot,
		quotRem,range,rangeSize,read,readDec,readFile,readFloat,readHex,readIO,readInt,readList,readLitChar,
		readLn,readOct,readParen,readSigned,reads,readsPrec,realToFrac,recip,rem,repeat,
		reverse,round,scaleFloat,scanl,scanl1,scanr,scanr1,seq,sequence,sequence_,show,showChar,showInt,
		showList,showLitChar,showParen,showSigned,showString,shows,showsPrec,significand,signum,sin,
		sinh,snd,span,splitAt,sqrt,subtract,succ,sum,tail,take,takeWhile,tan,tanh,threadToIOResult,toEnum,
		toInt,toInteger,toLower,toRational,toUpper,truncate,uncurry,undefined,unlines,until,unwords,unzip,
		unzip3,userError,words,writeFile,zip,zip3,zipWith,zipWith3,listArray,doParse,for,initTo,
                create,get,set,div,rescale,add,delete,insert,prop_focus_left_master,average,best,insert,union,split,size,fromList,copy,group,good,bad,foo,explode,singleton,difference,fromJust,sort,unfold,
                subst,unapply,apply,proxy,refinement,fresh,guard,constrain,oneOf,
                queryList,queryCtor,queryField,ctors,decodeCtor,whichOf,ctorArity,eval,
                mkCtor,gCtors,gEncode,gEncodeFields,gDecode,gDecodeFields,reproxyRep,empty,splitCtor,checkField,scanM,
                padAverage,focusUp,execute,checkSMT,inputTypes,outputType,toReft,app,
                genArgs,genWitness,close,witness,mkApps,mkLets,isStuck,findExpr,updState,putBefore,putAfter,
                putRoot,getNext,getPrev,getSubterms,applyCtx,findApp,findVal,replicate,loop,fac,incAllByOne,map
	},
	emphstyle={[1]\color{haskellblue}},
	emph=
	{[2]
		OkMap,OkRBT,OkStackSet,TTrue,Map,Bool,Char,Double,Either,Float,IO,Integer,Int,Maybe,Ordering,Rational,Ratio,ReadS,ShowS,String,Word8,Nat,Pos,Rng,Score,
                Ptr,ForeignPtr,CSize,InPacket,Tree,Prop,TreeEq,TreeLt,Vec,
                NullTerm,IncrList,DecrList,UniqList,BST,MinHeap,MaxHeap,
                PtrN,ByteStringN,ByteStringEq,VO,ByteStringsEq,ByteStringNE,OrdList,Var,RType,Constrain,Gen,Var,Proxy,SMT,Targetable,RefType,Refinement,Ctor,C1,Rep,Rec0,U1,
                GCtors,GDecode,GDecodeFields,GEncode,GEncodeFields,OrdMap,MinusKey,
                len,isBH,isBal,bh,isRB,keys,Sorted,RBT,Col,isBlack,OrdRBT,Set,sz,
                StackSet,NoDuplicates,Data,RBTree,XMonad,Generic,
                VState,Expr,Ctx,Cmd,int,list,string,bool
	},
	emphstyle={[2]\color{castanho_ulisses}},
	emph=
	{[3]
		case,class,data,deriving,do,else,if,return,def,import,in,infixl,infixr,instance,let,tmapM,for2M,forM,zipWithM,otherwise,
		module,measure,pred,predicate,of,primitive,then,type,where,lazy,throw,when,
                rec,function,fun,match,with
	},
	emphstyle={[3]\color{preto_ulisses}\textbf},
	emph=
	{[4]
		quot,rem,div,mod,elem,notElem,seq
	},
	emphstyle={[4]\color{castanho_ulisses}\textbf},
	emph=
	{[5]
		PS,Tip,Node,Black,Red,EQ,False,GT,Just,LT,Left,Nothing,Right,True,Show,Eq,Ord,Num,C,N,Leaf,Bin,CounterExample,
                StepForward,StepBack,JumpForward,JumpBack,StepOver,StepInto,true
	},
	emphstyle={[5]\color{green_ulisses}},
	emph=
	{[6]
		patError, irrefutPatError, nonExhaustiveGuardsError, recSelError, errorOut,
		noMethodBinding
	},
	emphstyle={[6]\color{haskellred}}
}

\lstdefinelanguage{HaskellUlissesMath}[]{HaskellUlisses}{mathescape=true}

\lstnewenvironment{code}
{\lstset{language=HaskellUlisses}}
{}

\lstnewenvironment{ecode}
{\lstset{language=HaskellUlisses
        ,numbers=left
        ,frame=leftline
        ,xleftmargin=7.5mm
        ,escapeinside={(*@}{@*)}
}}
{}

\lstnewenvironment{mcode}
{\lstset{language=HaskellUlissesMath}}
{}

\lstnewenvironment{fcode}
{\lstset{language=HaskellUlisses, frame=L}}
{}

\lstMakeShortInline[language=HaskellUlisses]@
\lstMakeShortInline[language=HaskellUlissesMath]|

\begin{document}
\maketitle

\begin{abstract}
  {Static} type errors are a common stumbling block
  for newcomers to typed functional languages.
  We present a {dynamic} approach to explaining type
  errors by generating counterexample witness inputs that
  illustrate how an ill-typed program goes wrong.
  First, given an ill-typed function, we symbolically
  execute the body to synthesize witness values that
  make the program go wrong.
  We prove that our procedure synthesizes
  {general} witnesses in that if a witness is
  found, then for all inhabited input types,
  there exist values that can make the function go wrong.
  Second, we show how to extend this procedure to
  produce a reduction graph that can be used to
  interactively visualize and debug witness executions.
  Third, we evaluate the coverage of our approach
  on two data sets comprising over 4,500 ill-typed
  student programs.
  Our technique is able to generate witnesses for
  around 85\% of the programs, our reduction graph
  yields small counterexamples for over 80\% of the witnesses,
  and a simple heuristic allows us to use witnesses to
  locate the source of type errors with around 70\% accuracy.
  Finally, we evaluate whether our witnesses help
  students understand and fix type errors, and
  find that students presented with our witnesses
  show a greater understanding of type errors
  than those presented with a standard error message.
\end{abstract}

\section{Introduction}
\label{sec:introduction}

Type errors are a common stumbling block for students
trying to learn typed functional languages like \ocaml\
and \haskell.
Consider the ill-typed @fac@ function on the left in
Figure~\ref{fig:factorial}.
The function returns @true@ in the base case (instead of @1@),
and so \ocaml responds with the error message:
\begin{verbatim}
  This expression has type
    bool
  but an expression was expected of type
    int.
\end{verbatim}
This message makes perfect sense to an expert who is familiar
with the language and has a good mental model of how the type
system works.
However, it may perplex a novice who has yet to develop such a
mental model.
To make matters worse, unification-based type inference algorithms
often report errors far removed from their source.
This further increases the novice's confusion and can actively mislead
them to focus their investigation on an irrelevant piece of code.
Much recent work has focused on analyzing unification constraints
to properly \emph{localize} a type error~\cite{Lerner2007-dt,Chen2014-gd,Zhang2014-lv,Pavlinovic2014-mr},
but an accurate source location does not explain \emph{why} the
program is wrong.

\begin{figure}[t]
\centering
\begin{minipage}{.49\linewidth}
\centering
\begin{ecode}
  let rec fac n =
    if n <= 0 then
      true
    else
      n * (*@\hlOcaml{fac (n-1)}@*)
\end{ecode}
\vspace{2em}
\includegraphics[height=1.5in]{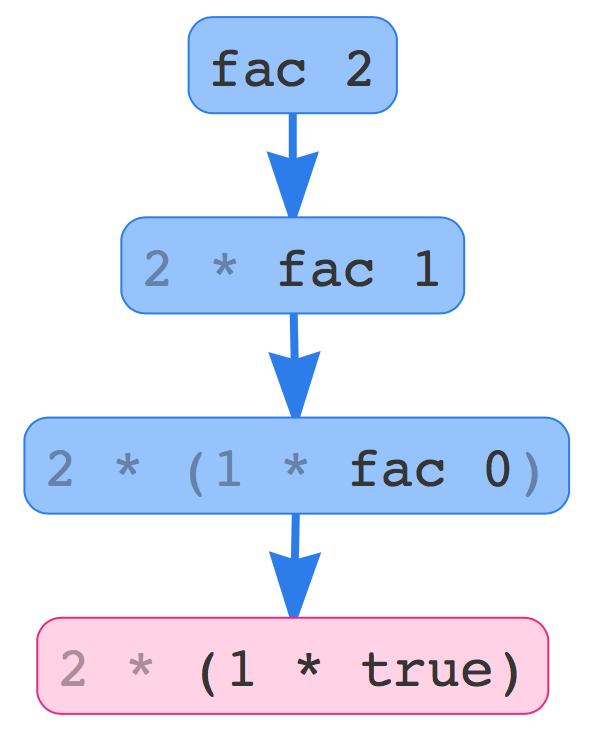}
\end{minipage}
\begin{minipage}{.49\linewidth}
\centering
\includegraphics[height=3in]{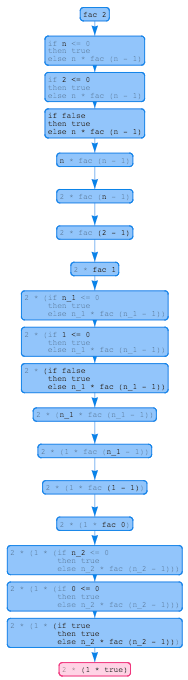}
\end{minipage}
\vspace{1em}
\caption{(top-left) An ill-typed \texttt{fac} function \hlOcaml{highlighting} the error location reported by \ocaml. (bottom-left) Dynamically witnessing the type error in \texttt{fac}, showing only function call-return pairs. (right) The same trace, fully expanded to show each small-step reduction in the computation.}
\label{fig:factorial}
\end{figure}

In this paper we propose a new approach that explains
static type errors by \emph{dynamically} witnessing
how an ill-typed program goes wrong.
We have developed \toolname, an interactive tool that uses
the source of the ill-typed function to automatically synthesize
the result on the bottom-left in Figure~\ref{fig:factorial}, which
shows how the recursive calls reduce to a configuration where
the program ``goes wrong'' --- \ie\ the @int@ value @1@ is to be
multiplied with the @bool@ value @true@.
We achieve this via three concrete contributions.

\paragraph{1. Finding Witnesses}
Our first contribution is an algorithm for searching for
\emph{witnesses} to type errors, \ie\ inputs that cause a
program to go wrong~(\S~\ref{sec:searching-witness}).
This problem is tricky when we cannot rely on
static type information, as we must avoid the
trap of \emph{spurious} inputs that cause
irrelevant problems that would be avoided
by picking values of a different, relevant type.
We solve this problem by developing a novel
operational semantics that combines evaluation
and type inference.
We execute the program with \emph{holes} --- values whose type is
unknown --- as the inputs.
A hole remains abstract until the evaluation
context tells us what type it must have, for
example the parameters to an addition operation
must both be integers.
Our semantics conservatively instantiates holes
with concrete values, dynamically inferring the
type of the input until the program goes wrong.
We prove that our procedure synthesizes \emph{general}
witnesses, which means, intuitively, that if a witness
is found for a given ill-typed function, then, \emph{for all}
(inhabited) input types, there exist values that can make
the function go wrong.

Given a witness to a type error, the novice may still be at a loss.
The standard \ocaml\ interpreter and debugging infrastructure expect
well-typed programs, so they cannot be used to investigate \emph{how}
the witness causes the program to crash.
More importantly, the execution itself may be quite long and may contain
details not relevant to the actual error.

\paragraph{2. Visualizing Witnesses}
Our second contribution is an interactive visualization of the
execution of purely functional \ocaml\ programs, well-typed or not~(\S~\ref{sec:interactive}).
We extend the semantics to also build a \emph{reduction graph}
which records all of the small-step reductions and the context
in which they occur.
The graph lets us visualize the sequence of
steps from the source witness to the stuck term. The user can
interactively expand the computation to expose intermediate steps
by selecting an expression and choosing a traversal strategy.
The strategies include many of the standard debugging moves, \eg\
stepping \emph{forward} or \emph{into} or \emph{over} calls, as well
stepping or jumping \emph{backward} to understand how a particular
value was created, while preserving a context of the intermediate
steps that allow the user to keep track of a term's provenance.

We introduce a notion of \emph{jump-compressed} traces to abstract away
the irrelevant details of a computation.
A jump-compressed trace includes only function
calls and returns. For example, the trace in the bottom-left of
Figure~\ref{fig:factorial} is jump-compressed.
Jump-compressed traces are similar to stack traces in that both show a
sequence of function calls that lead to a crash. However, jump-compressed
traces also show the return values of successful calls, which can be
useful in understanding why a particular path was taken.

\paragraph{3. Evaluating Witnesses}
Of course, the problem of finding witnesses is
undecidable in general. In fact, due to the necessarily
conservative nature of static typing, there
may not even exist any witnesses for a given
ill-typed program.
Thus, our approach is a heuristic that is only useful
if it can find \emph{compact} witnesses for
\emph{real-world} programs.
Our third contribution is an extensive evaluation of our approach
on two different sets of ill-typed programs obtained by instrumenting
compilers used in beginner's classes~(\S~\ref{sec:evaluation}).
The first is the \uwbench\ data set~\cite{Lerner2007-dt}
comprising \uwsize\ ill-typed programs.
The second is a new \ucsdbench\ data set, comprising \ucsdsize\
ill-typed programs.
We show that for both data sets, our technique is able to generate
witnesses for around 85\% of the programs, in under a second in the
vast majority of cases.
Furthermore, we show that a simple interactive strategy yields
compact counterexample traces with at most 5 steps for 60\%
of the programs, and at most 10 steps for over 80\% of the programs.
We can even use witnesses to \emph{localize} type errors with a simple
heuristic that treats the values in a ``stuck'' term as \emph{sources}
of typing constraints and the term itself as a \emph{sink},
achieving around 70\% accuracy in locating the source of the error.

The ultimate purpose of an error report is to help the programmer
\emph{comprehend} and \emph{fix} problematic code.
Thus, our final contribution is a user study that compares \toolname's
dynamic witnesses against \ocaml's type errors along the dimension of
comprehensibility~(\S~\ref{sec:user-study}).
Our study finds that students given one of our witnesses are
consistently more likely to correctly explain and fix a type
error than those given the standard error message produced by
the \ocaml compiler.

%
%
%
%
%
%

\smallskip
All together, our results show that in the vast majority of cases, (novices') ill-typed
programs \emph{do} go wrong, and that the witnesses to these errors can be
helpful in understanding the source of the error. This, in turn, opens the
door to a novel dynamic way to explain, understand, and appreciate the
benefits of static typing.

\paragraph{Contributions Relative to Prior Publications}
This paper extends our ICFP '16 paper of the same
name~\cite{Seidel2016-ul}, focusing on the experimental evaluation.
First, in \S~\ref{sec:how-safe} we investigate the student programs for
which we were unable to synthesize a witness. We group the failures into
five categories, give representative examples, and suggest ways to
improve our feedback in these cases. Interestingly, we find that in the
majority of these failed cases, the programs do not actually admit a
witness in our semantics.
Second, in \S~\ref{sec:locating} we attempt to use our witnesses to
localize type errors with a simple heuristic. We treat the stuck term as
a sink for typing constraints, and the values it contains as sources of
constraints. We can then predict that either the stuck term or one of
the terms that \emph{produced} a value it contains is likely at fault
for the error. We compare our localizations to \ocaml and two
state-of-the-art type error localization tools, and find that we are
competitive with the state of the art.
Finally, we have also extended \S~\ref{sec:user-study} with an analysis
of the statistical significance of our user study results.


\section{Overview}
\label{sec:overview}

We start with an overview of our approach to
explaining (static) type errors using \emph{witnesses}
that (dynamically) show how the program goes wrong.
We illustrate why generating suitable inputs
to functions is tricky in the absence of type
information.
Then we describe our solution to the problem
and highlight the similarity to static type
inference,
Finally, we demonstrate our visualization of
the synthesized witnesses.

\subsection{Generating Witnesses}
\label{sec:generating-witnesses}
Our goal is to find concrete values
that demonstrate how a program ``goes wrong''.

\paragraph{Problem: Which inputs are bad?}
One approach is to randomly generate input values and
use them to execute the program until we find one that
causes the program to go wrong.
Unfortunately, this approach quickly runs aground.
Recall the erroneous @fac@ function from Figure~\ref{fig:factorial}.
%
%
What \emph{types} of inputs should we test @fac@ with?
Values of type @int@ are fair game, but values of type, say,
@string@ or @int list@ will cause the program to go wrong
in an \emph{irrelevant} manner.
Concretely, we want to avoid testing @fac@ with any type other
than @int@ because any other type would cause @fac@ to get stuck
immediately in the @n <= 0@ test.

\paragraph{Solution: Don't generate inputs until forced.}
Our solution is to avoid generating a concrete value for the input at
all, until we can be sure of its type.
The intuition is that we want to be as lenient as possible in our tests,
so we make no assumptions about types until it becomes clear from the
context what type an input must have.
This is actually quite similar in spirit to type inference.

To defer input generation, we borrow the notion of a ``hole'' from
SmallCheck~\cite{Runciman2008-ka}.
A hole --- written \vhole{\thole} --- is a \emph{placeholder} for a
value \ehole of some unknown type \thole.
We leave all inputs as uninstantiated holes until they are demanded by
the program, \eg due to a primitive operation like the @<=@ test.

\paragraph{Narrowing Input Types}
Primitive operations, data construction, and case-analysis \emph{narrow}
the types of values.
For concrete values this amounts to a runtime type check, we ensure that
the value has a type compatible with the expected type.
For holes, this means we now know the type it should
have (or in the case of compound data we know \emph{more} about the
type) so we can instantiate the hole with a value.
The value may itself contain more holes, corresponding to components
whose type we still do not know.
Consider the @fst@ function:
\begin{code}
  let fst p = match p with
    (a, b) -> a
\end{code}
The case analysis tells us that @p@ must be a pair, but it says
\emph{nothing} about the contents of the pair.
Thus, upon reaching the case-analysis we would generate a pair
containing fresh holes for the @fst@ and @snd@ component.
Notice the similarity between instantiation of type variables and
instantiation of holes.
We can compute an approximate type for @fst@ by approximating the types
of the (instantiated) input and output, which would give us:
\begin{mcode}
  fst : ($\thole_1$ * $\thole_2$) -> $\thole_1$
\end{mcode}
We call this type approximate because we only see a single path through
the program, and thus will miss narrowing points that only occur in
other paths.

Returning to @fac@, given a hole as input we will narrow the hole
to an @int@ upon reaching the @<=@ test.
At this point we choose a
random @int@\footnote{With standard heuristics~\cite{Claessen2000-lj} to favor small values.}
for the instantiation and
concrete execution takes over entirely, leading us to the expected crash
in the multiplication.

\paragraph{Witness Generality}
We show in \S~\ref{sec:soundness} that our lazy instantiation of holes
produces \emph{general witnesses}.
That is, we show that if ``executing''
a function with a hole as input causes the
function to ``go wrong'', then there is
\emph{no possible} type for the function.
In other words, for \emph{any} types you might
assign to the function's inputs, there exist values
that will cause the function to go wrong.

\paragraph{Problem: How many inputs does a function take?}
There is another wrinkle, though; how did we know
that @fac@ takes a single argument instead of two
(or none)?
It is clear, syntactically, that @fac@ takes \emph{at least} one
argument, but in a higher-order language with currying, syntax can be
deceiving.
Consider the following definition:
\begin{code}
  let incAllByOne = List.map (+ 1)
\end{code}
Is @incAllByOne@ a function?
If so, how many arguments does it take?
The \ocaml\ compiler deduces that @incAllByOne@ takes a single argument
because the \emph{type} of \hbox{@List.map@} says it takes two arguments, and it is
partially applied to @(+ 1)@.
As we are dealing with ill-typed programs we do not have the luxury of
typing information.

\paragraph{Solution: Search for saturated application.}
We solve this problem by deducing the number of arguments
via an iterative process. We add arguments one-by-one
until we reach a \emph{saturated} application, \ie\
until evaluating the application returns a value
other than a lambda.

\subsection{Visualizing Witnesses}
\label{sec:visual-witness}
We have described how to reliably find witnesses to type errors in \ocaml,
but this does not fully address our original goal --- to \emph{explain}
the errors.
Having identified an input vector that triggers a crash, a common next
step is to step through the program with a \emph{debugger} to observe
how the program evolves.
The existing debuggers and interpreters for \ocaml\ assume a type-correct
program, so unfortunately we cannot use them off-the-shelf.
Instead we extend our search for witnesses to produce an execution
trace.

\paragraph{Reduction Graph}
Our trace takes the form of a reduction graph, which records small-step
reductions in the context in which they occur.
%
%
For example, evaluating the expression @1+2+3@ would produce the
graph in Figure~\ref{fig:simple-reduction-hi}.
\begin{figure}[t]
  \centering
  \includegraphics[height=2in]{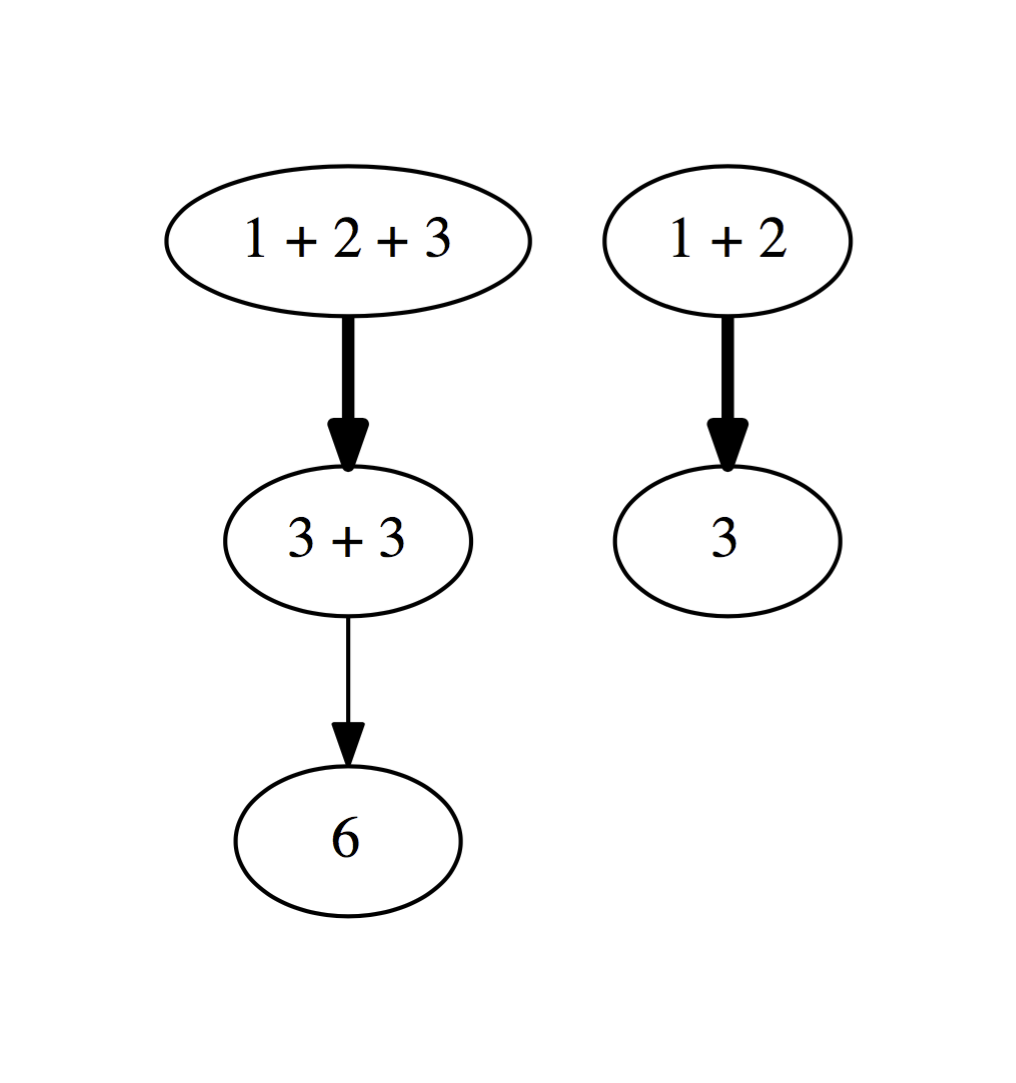}
  \caption{The reduction graph for \texttt{1+2+3}, highlighting the edges
    produced by reducing \texttt{1+2+3} to \hbox{\texttt{3+3}}.}
\label{fig:simple-reduction-hi}
\end{figure}
Notice that when we transition from @1+2+3@ to @3+3@ we collect
both that edge \emph{and} an edge from the sub-term @1+2@ to @3@.
These additional edges allow us to implement two common debugging
operations \emph{post-hoc}: ``step into'' to zoom in on a specific
function call, and ``step over'' to skip over uninteresting
computations.

\paragraph{Interacting with the graph}
The reduction graph is useful for formulating and executing traversals,
but displaying it all at once would quickly become overwhelming.
Our interaction begins by displaying a big-step reduction, \ie the
witness followed by the stuck term.
The user can then progressively fill in the hidden steps of the
computation by selecting a visible term and choosing one of the
applicable traversal strategies --- described in
\S~\ref{sec:interactive} --- to insert another term into the
visualization.

\paragraph{Jump-compressed Witnesses}
It is rare for the initial state of the visualization to be
informative enough to diagnose the error.
Rather than abandon the user, we provide a short-cut to expand the witness
to a \emph{jump-compressed} trace, which contains every function call
and return step.
The jump-compressed trace abstracts the computation as a sequence of
call-response pairs, providing a high-level overview of steps taken
to reach the crash, and a high level of compression compared to the
full trace.
For example, the jump-compressed trace in Figure~\ref{fig:factorial}
contains 4 nodes compared to the 19 in the fully expanded trace.
Our benchmark suite of student programs shows that jump-compression is
practical, with an average jump-compressed trace size of 7 nodes and a
median of 5.

\section{Type-Error Witnesses}
\label{sec:searching-witness}

Next, we formalize the notion of type error witnesses as follows.
First, we define a core calculus within which we will work~(\S~\ref{sec:syntax}).
Second, we develop a (non-deterministic) operational semantics
for ill-typed programs that precisely defines the notion
of a \emph{witness}~(\S~\ref{sec:semantics}).
Third, we formalize and prove a notion of \emph{generality} for
witnesses, which states, intuitively, that if we find a
single witness then for \emph{every possible} type
assignment there exist inputs that are guaranteed to make
the program ``go wrong''~(\S~\ref{sec:soundness}).
Finally, we refine the operational semantics into a
\emph{search procedure} that returns concrete (general)
witnesses for ill-typed programs~\S~(\ref{sec:search-algorithm}).
We have formalized and tested our semantics and generality theorem
in \textsc{PLT-Redex}~\cite{Felleisen2009-ya}.
Detailed proofs for the theorems in this section can be found in
\ifthenelse{\equal{\isTechReport}{true}}
{Appendix~\ref{sec:proofs}.}
{Appendix A of the accompanying tech report~\cite{Seidel2016Dynamic-TechRep}.}

\subsection{Syntax}
\label{sec:syntax}
\begin{figure}
$$
\begin{array}{rrcl}

\emphbf{Expressions}
  & \estuck & ::= & e \spmid \stuck \\
  & e & ::=    & v \spmid x \spmid \eapp{e}{e} \spmid \eplus{e}{e}\\
  &   & \spmid & \eif{e}{e}{e} \\
  &   & \spmid & \epair{e}{e} \spmid \epcase{e}{x}{x}{e} \\
  &   & \spmid & \enode{e}{e}{e} \spmid \eleaf \\
  &   & \spmid & \ecase{e}{e}{x}{x}{x}{e} \\[0.1in]

\emphbf{Values}
  & v  & ::= & n \spmid b \spmid \efun{x}{e} \spmid \vhole{\thole} \spmid \epair{v}{v} \spmid tr \\
  & tr & ::= & \vnode{t}{v}{v}{v} \spmid \vleaf{t} \\[0.05in]

\emphbf{Integers}
  & n & ::= &  0,1,-1,\ldots \\[0.05in]

\emphbf{Booleans}
  & b & ::= &  \etrue \spmid \efalse \\[0.05in]

\emphbf{Types}
  & t & ::=     & \tbool \spmid \tint \spmid \tfun \\
  &   &  \spmid & \tprod{t}{t} \spmid \ttree{t} \spmid \thole \\[0.05in]

\emphbf{Substitutions}
  & \vsu & ::= & \emptysu \spmid \extendsu{\vsu}{\vhole{\thole}}{v} \\
  & \tsu & ::= & \emptysu \spmid \extendsu{\tsu}{\thole}{t} \\[0.1in]
\emphbf{Contexts}
  & C
  & ::=
  &   	 \bullet
  \spmid \eapp{C}{e}
  \spmid \eapp{v}{C} \\
  & & \spmid & \eplus{C}{e} \spmid \eplus{v}{C} \\
  & & \spmid & \eif{C}{e}{e} \\
  & & \spmid & \epair{C}{e} \spmid \epair{v}{C} \\
  & & \spmid & \epcase{C}{x}{x}{e} \\
  & & \spmid & \enode{C}{e}{e} \spmid \enode{v}{C}{e} \spmid \enode{v}{v}{C} \\
  & & \spmid & \ecase{C}{e}{x}{x}{x}{e} \\[0.05in]
\end{array}
$$



\caption{Syntax of \lang}
\label{fig:syntax}
\end{figure}

Figure~\ref{fig:syntax} describes the syntax of \lang, a simple lambda
calculus with integers, booleans, pairs, and binary trees.
As we are specifically interested in programs that \emph{do} go wrong,
we include an explicit \stuck\ term in our syntax. We write $\estuck$ to
denote terms that may be \stuck, and $e$ to denote terms that may not be
stuck.

\paragraph{Holes}
\label{sec:holes}
Recall that a key challenge in our setting is to find witnesses
that are meaningful and do not arise from choosing values from
irrelevant types.
We solve this problem by equipping our term language with a
notion of a \emph{hole}, written \vhole{\thole}, which represents
an \emph{unconstrained} value $\ehole$ that may be replaced with
\emph{any} value of an unknown type \thole.
Intuitively, the type holes \thole\ can be viewed as type variables
that we will \emph{not} generalize over.
A \emph{normalized} value is one that is not a hole,
but which may internally contain holes.
For example
$\vnode{\thole}{\vhole{\thole}}{\vleaf{\thole}}{\vleaf{\thole}}$ is a
normalized value.

\paragraph{Substitutions}
Our semantics ensure the generality of witnesses by incrementally
\emph{refining} holes, filling in just as much information as is
needed locally to make progress (inspired by the manner in
which SmallCheck uses lazy evaluation~\cite{Runciman2008-ka}).
We track how the holes are incrementally filled in, by using
value (resp.\ type) \emph{substitutions} $\vsu$ (resp. $\tsu$)
that map value (resp.\ type) holes to values (resp.\ types).
The substitutions let us ensure that we consistently instantiate each
hole with the same (partially defined) value or type, regardless of the
multiple contexts in which the hole appears.
This ensures we can report a concrete (and general) witness for any
(dynamically) discovered type errors.

A \emph{normalized} value substitution is one whose co-domain is
comprised of normalized values.
In the sequel, we will assume and ensure that all value substitutions
are normalized.
We ensure additionally that the co-domain of a substitution does not
refer to any elements of its domain, \ie when we extend a substitution
with a new binding we apply the substitution to itself.


\subsection{Semantics}
\label{sec:semantics}
Recall that our goal is to synthesize a value that demonstrates
why (and how) a function goes wrong.
We accomplish this by combining evaluation with type inference,
giving us a form of dynamic type inference.
%
Each primitive evaluation step tells us more about the types of the
program values. For example, addition tells us that the addends must be
integers, and 
an if-expression tells us the condition must be a boolean.
When a hole appears in such a context, we know what type it must have
in order to make progress and can fill it in with a concrete value.

The evaluation relation is parameterized by a pair of functions,
\emph{narrow} (\forcesym) and \emph{generate} (\gensym),
that ``dynamically'' perform type-checking and hole-filling
respectively.

\paragraph{Narrowing Types}
The procedure 
\[
\forcesym : v \times t \times \vsu \times \tsu \rightarrow \triple{v \cup \stuck}{\vsu}{\tsu}
\]%
defined in Figure~\ref{fig:narrow}, takes as input a value $v$, a type
$t$, and the current value and type substitutions, and refines $v$ to
have type $t$ by yielding a triple of either the same value and
substitutions, or yields the stuck state if no such refinement is
possible. In the case where $v$ is a hole, it first checks in the given
$\vsu$ to see if the hole has already been instantiated and, if so,
returns the existing instantiation.
\begin{figure}[t]
\[\def\arraystretch{1.5}
\begin{array}{lll}
\forcesym                  & : v \times t \times \vsu \times \tsu \rightarrow \triple{v \cup \stuck}{\vsu}{\tsu} \\

\force{\vhole{\thole}}{t}{\vsu}{\tsu} & \defeq 
\begin{cases}
  \triple{v}{\vsu}{\tsu'}    &\ \textbf{if} \begin{array}{l} v = \lookupsu{\vsu}{\vhole{\thole}}, \\
                                         \tsu' = \unify{\{\thole, t, \typeof{v}\}}{\tsu}\end{array} \\[0.16in]
  \triple{\stuck}{\vsu}{\tsu} &\ \textbf{if} \begin{array}{l} v = \lookupsu{\vsu}{\vhole{\thole}}\end{array} \\[0.08in]
  \triple{v}{\extendsu{\vsu}{\vhole{\thole}}{v}}{\tsu'} &\ \textbf{if} \begin{array}{l} \tsu' = \unify{\{\thole, t\}}{\tsu}, \\ v = \gen{t}{\tsu'} \end{array}\\
\end{cases}  \\ 
\force{n}{\tint}{\vsu}{\tsu}     & \defeq \triple{n}{\vsu}{\tsu} \\
\force{b}{\tbool}{\vsu}{\tsu}    & \defeq \triple{b}{\vsu}{\tsu} \\
\force{\efun{x}{e}}{\tfun}{\vsu}{\tsu} & \defeq \triple{\efun{x}{e}}{\vsu}{\tsu} \\
\force{\epair{v_1}{v_2}}{\tprod{t_1}{t_2}}{\vsu}{\tsu} & \defeq \triple{\epair{v_1}{v_2}}{\vsu}{\tsu''} & \hspace{-43.5mm} \textbf{if} \begin{array}{l} \tsu' = \unify{\{\typeof{v_1}, t_1\}}{\tsu},\\ \tsu'' = \unify{\{\typeof{v_2}, t_2\}}{\tsu'} \end{array} \\
\force{\vleaf{t_1}}{\ttree{t_2}}{\vsu}{\tsu} & \defeq \triple{\vleaf{t_1}}{\vsu}{\tsu'} & \hspace{-43.5mm} \textbf{if} \begin{array}{l} \tsu' = \unify{\{t_1, t_2\}}{\tsu} \end{array}\\
\force{\vnode{t_1}{v_1}{v_2}{v_3}}{\ttree{t_2}}{\vsu}{\tsu} \hspace{-3mm} & \defeq \triple{\vnode{t_1}{v_1}{v_2}{v_3}}{\vsu}{\tsu'} & \hspace{-43.5mm} \textbf{if} \begin{array}{l} \tsu' = \unify{\{t_1, t_2\}}{\tsu} \end{array} \\
\force{v}{t}{\vsu}{\tsu} & \defeq \triple{\stuck}{\vsu}{\tsu}
\end{array}
\]
\caption{Narrowing values}
\label{fig:narrow}
\end{figure}
%
As the value substitution is normalized, in the first case of \forcesym\ we
do not need to \forcesym\ the result of the substitution, the sub-hole
will be narrowed when the context demands it.

\paragraph{Generating Values} The (non-deterministic)
$\gen{t}{\tsu}$ in Figure~\ref{fig:gen} takes
as input a type $t$ and returns a value of that type.
For base types the procedure returns an arbitrary value of
that type.
For functions it returns a lambda with a \emph{new} hole
denoting the return value.
For unconstrained types (denoted
by $\thole$) it yields a fresh hole constrained to have type
\thole (denoted by $\vhole{\thole}$).
When generating a $\ttree{t}$ we must take care to ensure
the resulting tree is well-typed.
For a polymorphic type $\ttree{\thole}$ 
or $\tprod{\thole_1}{\thole_2}$
we will place holes in the generated value; they will be lazily filled
in later, on demand.

\begin{figure}[t]
\[
\begin{array}{lcll}
\gensym       & :   & t \times \tsu \rightarrow v \\
\gen{\thole}{\tsu}  & \defeq  & \gen{\subst{\tsu}{\thole}}{\tsu} &  \text{if } \thole \in dom(\tsu) \\
\gen{\tint}{\tsu}   & \defeq  & n &  \text{non-det.} \\
\gen{\tbool}{\tsu}  & \defeq  & b &  \text{non-det.} \\
\gen{\tprod{t_1}{t_2}}{\tsu}  & \defeq  & \epair{\gen{t_1}{\tsu}}{\gen{t_2}{\tsu}} & \\ 
\gen{\ttree{t}}{\tsu}  & \defeq  & tr &  \text{non-det.} \\
\gen{\tfun}{\tsu}   & \defeq & \efun{x}{\vhole{\thole}} &  \text{\ehole, \thole are fresh} \\
\gen{\thole}{\tsu}  & \defeq & \vhole{\thole} & \text{\ehole is fresh} \\
\end{array}
\]
\caption{Generating values}
\label{fig:gen}
\end{figure}

\paragraph{Steps and Traces}
\begin{figure}[t]
\judgementHead{Evaluation}{\step{\estuck}{\vsu}{\tsu}{\estuck}{\vsu}{\tsu}}
\centerline{
\begin{minipage}{1.2\textwidth}
\begin{gather*}
\inference[\replusgood]
  {\triple{n_1}{\vsu'}{\tsu'} = \force{v_1}{\tint}{\vsu}{\tsu} \\
   \triple{n_2}{\vsu''}{\tsu''} = \force{v_2}{\tint}{\vsu'}{\tsu'} \\
   n = \eplus{n_1}{n_2}}
  {\step{\inctx{\eplus{v_1}{v_2}}}{\vsu}{\tsu}{\inctx{n}}{\vsu''}{\tsu''}}
\quad
\inference[\replusbadone]
  {\triple{\stuck}{\vsu'}{\tsu'} = \force{v_1}{\tint}{\vsu}{\tsu}}
  {\step{\inctx{\eplus{v_1}{v_2}}}{\vsu}{\tsu}{\stuck}{\vsu'}{\tsu'}}
\\[0.05in]
\inference[\replusbadtwo]
  {\triple{\stuck}{\vsu'}{\tsu'} = \force{v_2}{\tint}{\vsu}{\tsu}}
  {\step{\inctx{\eplus{v_1}{v_2}}}{\vsu}{\tsu}{\stuck}{\vsu'}{\tsu'}}
\quad
\inference[\reifgoodone]
  {\triple{\etrue}{\vsu'}{\tsu'} = \force{v}{\tbool}{\vsu}{\tsu}}
  {\step{\inctx{\eif{v}{e_1}{e_2}}}{\vsu}{\tsu}{\inctx{e_1}}{\vsu'}{\tsu'}}
\\[0.05in]
\inference[\reifgoodtwo]
  {\triple{\efalse}{\vsu'}{\tsu'} = \force{v}{\tbool}{\vsu}{\tsu}}
  {\step{\inctx{\eif{v}{e_1}{e_2}}}{\vsu}{\tsu}{\inctx{e_2}}{\vsu'}{\tsu'}}
\quad
\inference[\reifbad]
  {\triple{\stuck}{\vsu'}{\tsu'} = \force{v}{\tbool}{\vsu}{\tsu}}
  {\step{\inctx{\eif{v}{e_1}{e_2}}}{\vsu}{\tsu}{\stuck}{\vsu'}{\tsu'}}
\\[0.05in]
\inference[\reappgood]
  {\triple{\efun{x}{e}}{\vsu'}{\tsu'} = \force{v_1}{\tfun}{\vsu}{\tsu}}
  {\step{\inctx{\eapp{v_1}{v_2}}}{\vsu}{\tsu}{\inctx{e\sub{x}{v_2}}}{\vsu'}{\tsu'}}
\quad
\inference[\reappbad]
  {\triple{\stuck}{\vsu'}{\tsu'} = \force{v_1}{\tfun}{\vsu}{\tsu}}
  {\step{\inctx{\eapp{v_1}{v_2}}}{\vsu}{\tsu}{\stuck}{\vsu'}{\tsu'}}
\\[0.05in]
\inference[\releafgood]
  {\thole \mbox{ is fresh}}
  {\step{\inctx{\eleaf}}{\vsu}{\tsu}{\inctx{\vleaf{\thole}}}{\vsu}{\tsu}}
\quad
\inference[\renodegood]
  {
   t = \typeof{v_1} \\
   \triple{v_2'}{\vsu_2}{\tsu_2} = \force{v_2}{\ttree{t}}{\vsu_1}{\tsu_1} \\
   \triple{v_3'}{\vsu_3}{\tsu_3} = \force{v_3}{\ttree{t}}{\vsu_2}{\tsu_2} \\
  }
  {\step{\inctx{\enode{v_1}{v_2}{v_3}}}{\vsu}{\tsu}
        {\inctx{\vnode{t}{v_1}{v_2'}{v_3'}}}{\vsu_3}{\tsu_3}}
\\[0.05in]
\inference[\renodebadone]
  {
   t = \typeof{v_1} \\
   \triple{\stuck}{\vsu_2}{\tsu_2} = \force{v_2}{\ttree{t}}{\vsu_1}{\tsu_1} \\
  }
  {\step{\inctx{\enode{v_1}{v_2}{v_3}}}{\vsu}{\tsu}
        {\stuck}{\vsu_3}{\tsu_3}}
\quad
\inference[\renodebadtwo]
  {
   t = \typeof{v_1} \\
   \triple{v_2'}{\vsu_2}{\tsu_2} = \force{v_2}{\ttree{t}}{\vsu_1}{\tsu_1} \\
   \triple{\stuck}{\vsu_3}{\tsu_3} = \force{v_3}{\ttree{t}}{\vsu_2}{\tsu_2} \\
  }
  {\step{\inctx{\enode{v_1}{v_2}{v_3}}}{\vsu}{\tsu}
        {\stuck}{\vsu_3}{\tsu_3}}
\\[0.05in]
\inference[\recasegoodone]
  {\thole \mbox{ is fresh} & \triple{\vleaf{t}}{\vsu_1}{\tsu_1} = \force{v}{\ttree{\thole}}{\vsu}{\tsu}
  }
  {\step{\inctx{\ecase{v}{e_1}{x_1}{x_2}{x_3}{e_2}}}{\vsu}{\tsu}
        {\inctx{e_1}}{\vsu_1}{\tsu_1}}
\\[0.05in]
\inference[\recasegoodtwo]
  {\thole \mbox{ is fresh} & \triple{\vnode{t}{v_1}{v_2}{v_3}}{\vsu_1}{\tsu_1} = \force{v_1}{\ttree{\thole}}{\vsu}{\tsu}
  }
  {\step{\inctx{\ecase{v}{e_1}{x_1}{x_2}{x_3}{e_2}}}{\vsu}{\tsu}
        {\inctx{e_2\sub{x_1}{v_1}\sub{x_2}{v_2}\sub{x_3}{v_3}}}{\vsu_1}{\tsu_1}}
\\[0.05in]
\inference[\recasebad]
  {\thole \mbox{ is fresh} & \triple{\stuck}{\vsu_1}{\tsu_1} = \force{v}{\ttree{\thole}}{\vsu}{\tsu}
  }
  {\step{\inctx{\ecase{v}{e_1}{x_1}{x_2}{x_3}{e_2}}}{\vsu}{\tsu}
        {\stuck}{\vsu_1}{\tsu_1}}
\\[0.05in]
\inference[\repcasegood]
  {\thole_1, \thole_2 \mbox{ are fresh} & \triple{\epair{v_1}{v_2}}{\vsu_1}{\tsu_1} = \force{v}{\tprod{\thole_1}{\thole_2}}{\vsu}{\tsu}
  }
  {\step{\inctx{\epcase{v}{x_1}{x_2}{e}}}{\vsu}{\tsu}
        {\inctx{e\sub{x_1}{v_1}\sub{x_2}{v_2}}}{\vsu_1}{\tsu_1}}
\\[0.05in]
\inference[\repcasebad]
  {\thole_1, \thole_2 \mbox{ are fresh} & \triple{\stuck}{\vsu_1}{\tsu_1} = \force{v}{\tprod{\thole_1}{\thole_2}}{\vsu}{\tsu}
  }
  {\step{\inctx{\epcase{v}{x_1}{x_2}{e}}}{\vsu}{\tsu}
        {\stuck}{\vsu_1}{\tsu_1}}
\\[0.05in]
\end{gather*}
\end{minipage}
}
\caption{Evaluation relation for \lang}
\label{fig:operational}
\end{figure}

%
Figure~\ref{fig:operational} describes the small-step contextual
reduction semantics for \lang.
A configuration is a triple $\triple{\estuck}{\vsu}{\tsu}$ of an
expression $e$ or the stuck term $\stuck$, a value substitution $\vsu$,
and a type substitution $\tsu$.
%
%
We write $\step{\estuck}{\vsu}{\tsu}{\estuck'}{\vsu'}{\tsu'}$ if the state
$\triple{\estuck}{\vsu}{\tsu}$ transitions in a \emph{single step} to
$\triple{\estuck'}{\vsu'}{\tsu'}$.
A (finite) \emph{trace} $\trace$ is a sequence of configurations
$\triple{\estuck_0}{\vsu_0}{\tsu_0}, \ldots, \triple{\estuck_n}{\vsu_n}{\tsu_n}$ such that
$\forall 0 \leq i < n$, we have
$\step{\estuck_i}{\vsu_i}{\tsu_i}{\estuck_{i+1}}{\vsu_{i+1}}{\tsu_{i+1}}$.
We write $\steptr{\trace}{\estuck}{\vsu}{\tsu}{\estuck'}{\vsu'}{\tsu'}$ if $\trace$ is
a trace of the form $\triple{\estuck}{\vsu}{\tsu},\ldots,$ $\triple{\estuck'}{\vsu'}{\tsu'}$.
We write \steps{\estuck}{\vsu}{\tsu}{\estuck'}{\vsu'}{\tsu'} if
\steptr{\trace}{\estuck}{\vsu}{\tsu}{\estuck'}{\vsu'}{\tsu'} for some trace $\trace$.

\paragraph{Primitive Reductions}
%
%
Primitive reduction steps --- addition, if-elimination,
function application, and data construction and case
analysis --- use \forcesym to ensure that values have
the appropriate type (and that holes are instantiated)
before continuing the computation.
Importantly, beta-reduction \emph{does not} type-check its
argument, it only ensures that ``the caller'' $v_1$ is indeed
a function.

\paragraph{Recursion}
%
Our semantics lacks a built-in $\mathtt{fix}$ construct
for defining recursive functions, which may surprise
the reader.
Fixed-point operators often cannot be typed in static type
systems, but our system would simply approximate its type
as $\tfun$, apply it, and move along with evaluation.
Thus we can use any of the standard fixed-point operators
and do not need a built-in recursion construct. 




\subsection{Generality}\label{sec:soundness}

A key technical challenge in generating witnesses is
that we have no (static) type information to rely upon.
Thus, we must avoid the trap of generating \emph{spurious}
witnesses that arise from picking irrelevant values, when
instead there exist perfectly good values of a \emph{different}
type under which the program would not have gone wrong.
We now show that our evaluation relation instantiates holes
in a \emph{general} manner. That is, given a lambda-term $f$,
if we have $\steps{\eapp{f}{\vhole{\thole}}}{\emptysu}{\emptysu}{\stuck}{\vsu}{\tsu}$,
then \emph{for every} concrete type $t$, we can find a value
$v$ of type $t$ such that $\eapp{f}{v}$ goes wrong.

\begin{thm}[Witness Generality]
\label{thm:soundness}
  For any lambda-term $f$, if
  \hbox{$\steptr{\trace}{\eapp{f}{\vhole{\thole}}}{\emptysu}{\emptysu}{\stuck}{\vsu}{\tsu}$,}
  then for every
  (inhabited\footnote{All types in \lang are inhabited, but in a larger language like \ocaml this may not be true.})
  type
  $t$ there exists a value $v$ of type $t$ such that
  $\steps{\eapp{f}{v}}{\emptysu}{\emptysu}{\stuck}{\vsu'}{\tsu'}$.
\end{thm}

We need to develop some machinery in order to prove this theorem.
First, we show how our evaluation rules encode a dynamic form of
type inference, and then we show that the witnesses found by
evaluation are indeed maximally general.

\paragraph{The Type of a Value} The \emph{dynamic type}
of a value $v$ is defined as a function $\typeof{v}$ shown
in Figure~\ref{fig:typeof}.
The types of primitive values are defined in the natural manner.
The types of functions are \emph{approximated}, which is all
that is needed to ensure an application does not get stuck.
For example,
$$\typeof{\efun{x}{\eplus{x}{1}}} = \tfun$$
instead of $\tint \rightarrow \tint$.
The types of tuples are obtained directly from their values, and the
types of (polymorphic) trees from their labels.
Note that the evaluation relation in Figure~\ref{fig:operational}
guarantees that tree \emph{values} will be annotated with their type.
For nodes the type can be taken from the type of its value $v_1$, but
for leaves the evaluation relation creates a new type hole $\thole$
(this corresponds to polymorphic instantiation in a typechecker).

\begin{figure}[t]
\[ \begin{array}{lcll}
    \typeof{n}   & \defeq & \tint & \\
    \typeof{b}   & \defeq & \tbool & \\
    \typeof{\efun{x}{e}} & \defeq & \tfun \\
    \typeof{\epair{v_1}{v_2}} & \defeq & \tprod{\typeof{v_1}}{\typeof{v_2}} \\
    \typeof{\vleaf{t}} & \defeq & \ttree{t} \\
    \typeof{\vnode{t}{v_1}{v_2}{v_3}} & \defeq & \ttree{t} \\
    \typeof{\vhole{\thole}} & \defeq & \thole \\
  \end{array} \]
\caption{The \emph{dynamic type} of a value.}
\label{fig:typeof}
\end{figure}

\paragraph{Dynamic Type Inference}
We can think of the evaluation of \eapp{f}{\vhole{\thole}}
as synthesizing a partial instantiation of \thole, and thus
\emph{dynamically inferring} a (partial) type for $f$'s input.
We can extract this type from an evaluation trace by
applying the final type substitution to \thole.
%
Formally, we say that if
$\steptr{\trace}{\eapp{f}{\vhole{\thole}}}{\emptysu}{\emptysu}{\estuck}{\vsu}{\tsu}$,
then the \emph{partial input type} of $f$ up to $\trace$, written
\ptype{\trace}{f}, is $\resolve{\thole}{\tsu}$.

\paragraph{Compatibility}
A type $s$ is \emph{compatible} with a type $t$, written \tcompat{s}{t},
if $\exists \tsu.\ \subst{\tsu}{s} = \subst{\tsu}{t}$.
That is, two types are compatible if there exists a type substitution
that maps both types to the same type.
A value $v$ is \emph{compatible} with a type $t$, written \vcompat{v}{t},
if $\tcompat{\typeof{v}}{t}$, that is, if the dynamic type of $v$ is
compatible with $t$.


\paragraph{Type Refinement}
A type $s$ is a \emph{refinement} of a type $t$, written $\tsub{s}{t}$,
if $\exists \theta. s = \subst{\theta}{t}$.
In other words, $s$ is a refinement of $t$ if there exists a type
substitution that maps $t$ directly to $s$.
A type $t$ is a \emph{refinement} of a value $v$, written $\tsub{t}{v}$,
if $\tsub{t}{\typeof{v}}$, \ie if $t$ is a refinement of the
dynamic type of $v$.
%


\paragraph{Preservation}
We prove two preservation lemmas. First, we show that each evaluation
step refines the partial input type of $f$, thus preserving type
compatibility.
\begin{lem}
\label{lem:vsu-ext}
If  $\trace \defeq \triple{\eapp{f}{\vhole{\thole}}}{\emptysu}{\emptysu},\ldots,\triple{e}{\vsu}{\tsu}$
and $\trace' \defeq \trace, \step{e}{\vsu}{\tsu}{e'}{\vsu'}{\tsu'}$
  (\ie\ $\trace'$ is a single-step extension of $\trace$)
and $\ptype{\trace}{f} \neq \ptype{\trace'}{f}$
then $\tsu' = \tsu[\thole_1 \mapsto t_1] \ldots [\thole_n \mapsto t_n]$.
\end{lem}
\begin{proof}
  By case analysis on the evaluation rules.
  $\thole$ does not change, so if the partial input types differ then
  $\tsu \neq \tsu'$.
  Only \forcesym\ can change \tsu,
  via \unifysym, which can only extend \tsu.
  %
\end{proof}
Second, we show that at each step of evaluation, the partial input type of $f$
is a refinement of the instantiation of $\vhole{\thole}$.
\begin{lem}
\label{lem:resolve-compat}
For all traces
$\trace \defeq \triple{\eapp{f}{\vhole{\thole}}}{\emptysu}{\emptysu},\ldots,\triple{e}{\vsu}{\tsu}$,
$\vsub{\ptype{\trace}{f}}{\resolve{\vhole{\thole}}{\vsu}}$.
\end{lem}
\begin{proof}
  By induction on $\trace$.
  In the base case $\trace = \triple{\eapp{f}{\vhole{\thole}}}{\emptysu}{\emptysu}$
  and $\thole$ trivially refines $\vhole{\thole}$.
  In the inductive case, consider the single-step extension of $\trace$,
  $\trace' = \trace,\triple{e'}{\vsu'}{\tsu'}$.
  We show by case analysis on the evaluation rules that if
  $\vsub{\ptype{\trace}{f}}{\resolve{\vhole{\thole}}{\vsu}}$, then
  $\vsub{\ptype{\trace'}{f}}{\resolve{\vhole{\thole}}{\vsu'}}$.
\end{proof}

\paragraph{Incompatible Types Are Wrong}
\emph{For all} types that are \emph{incompatible} with the
partial input type up to $\trace$, there exists a value
that will cause $f$ to get stuck in \emph{at most} $k$ steps,
where $k$ is the length of $\trace$.

\begin{lem}
\label{lem:k-stuck}
For all types $t$,
if $\steptr{\trace}{\eapp{f}{\vhole{\thole}}}{\emptysu}{\emptysu}{e}{\vsu}{\tsu}$ and
   $\tincompat{t}{\ptype{\trace}{f}}$,
   then there exists a $v$ such that $\hastype{v}{t}$ and
   $\steps{\eapp{f}{v}}{\emptysu}{\emptysu}{\stuck}{\vsu'}{\tsu'}$ in at most
   $k$ steps, where $k$ is the length of $\trace$.
\end{lem}
\begin{proof}
  We can construct $v$ from $\trace$ as follows.
  Let
  $$
  \trace_i = \triple{\eapp{f}{\vhole{\thole}}}{\emptysu}{\emptysu},
             \ldots,
             \triple{e_{i-1}}{\vsu_{i-1}}{\tsu_{i-1}},
             \triple{e_{i}}{\vsu_{i}}{\tsu_{i}}
  $$
  be the shortest prefix of $\trace$ such that
  $\tincompat{\ptype{\trace_i}{f}}{t}$.
  We will show that $\ptype{\trace_{i-1}}{f}$ 
  must contain some other hole $\thole'$ that is
  instantiated at step $i$.
  Furthermore, $\thole'$ is instantiated in such a way that
  $\tincompat{\ptype{\trace_i}{f}}{t}$.
  %
  Finally, we will show that if we had instantiated $\thole'$ such that
  $\tcompat{\ptype{\trace_i}{f}}{t}$,
  the current step would have gotten $\stuck$.


  %
  By Lemma~\ref{lem:vsu-ext} we know that
  $\tsu_{i} = \tsu_{i-1}[\thole_1 \mapsto t_1] \ldots [\thole_n \mapsto t_n]$.
  We will assume, without loss of generality, that
  $\tsu_{i} = \tsu_{i-1}[\thole' \mapsto t']$.
  Since $\tsu_{i-1}$ and $\tsu_{i}$ differ only in $\thole'$ but the resolved
  types differ, we have
  $\thole' \in \ptype{\trace_{i-1}}{f}$
  and
  $\ptype{\trace_i}{f} = \ptype{\trace_{i-1}}{f}\sub{\thole'}{t'}$.
  %
  Let $s$ be a
  concrete type such that $\ptype{\trace_{i-1}}{f}\sub{\thole'}{s} = t$.
  We show by case analysis on the evaluation rules that
  $$\step{e_{i-1}}{\vsu_{i-1}}{\tsu_{i-1}[\thole' \mapsto s]}{\stuck}{\vsu}{\tsu}$$

  Finally, by Lemma~\ref{lem:resolve-compat} we know that
  $\vsub{\ptype{\trace_{i-1}}{f}}{\resolve{\vhole{\thole}}{\vsu_{i-1}}}$
  and thus $\thole' \in \resolve{\vhole{\thole}}{\vsu_{i-1}}$.
  %
  Let
  $$
  \begin{array}{lcl}
  u &=& \gen{s}{\tsu} \\
  v &=& \resolve{\vhole{\thole}}{\vsu_{i-1}}\sub{\ehole'[\thole']}{u}\sub{\thole'}{s} \\
  \end{array}
  $$
  $\steps{\eapp{f}{v}}{\emptysu}{\emptysu}{\stuck}{\vsu}{\tsu}$ in $i$ steps.
  %
  %
  %
\end{proof}

\begin{proof}[\textbf{Proof of Theorem~\ref{thm:soundness}}]
Suppose $\trace$ witnesses that $f$ gets stuck,
and let $s = \ptype{\trace}{f}$.
We show that \emph{all} types $t$ have stuck-inducing
values by splitting cases on whether $t$ is
compatible with $s$. 
\begin{description}
\item [Case \tcompat{s}{t}:]
  Let $\trace = \triple{\eapp{f}{\vhole{\thole}}}{\emptysu}{\emptysu},\ldots,\triple{\stuck}{\vsu}{\tsu}$.
  The value $v = \resolve{\vhole{\thole}}{\vsu}$ demonstrates that
  $\eapp{f}{v}$ gets stuck.
\item [Case \tincompat{s}{t}:] By Lemma~\ref{lem:k-stuck}, we can derive
  a $v$ from $\trace$ such that $\hastype{v}{t}$ and $\eapp{f}{v}$ gets
  stuck.
\end{description}
\end{proof}

\subsection{Search Algorithm}
\label{sec:search-algorithm}
So far, we have seen how a trace leading to a stuck configuration yields
a general witness demonstrating that the program is ill-typed (\ie\ goes
wrong for at least one input of every type).
In particular, we have shown how to non-deterministically find a witnesses
for a function of a \emph{single} argument.

We must address two challenges
to convert the semantics into a \emph{procedure} for finding
witnesses.
First, we must resolve the non-determinism introduced by \gensym.
Second, in the presence of higher-order functions and currying,
we must determine how many concrete values to generate to make
execution go wrong (as we cannot rely upon static typing to
provide this information.)

The witness generation procedure $\genWitnessN$ is formalized in
Figure~\ref{fig:algo-gen-witness}.
Next, we describe its input and output, and how it
addresses the above challenges to search the space of possible
executions for general type error witnesses.

\paragraph{Inputs and Outputs}
The problem of generating inputs is undecidable in general.
Our witness generation procedure takes two inputs:
(1) a search bound $k$ which is used to define the \emph{number} of
traces to explore\footnote{We assume, without loss of generality, that all
traces are finite.} and
(2) the target expression $e$ that contains the type error
(which may be a curried function of multiple arguments).
The witness generation procedure returns a list of (general)
witness expressions, each of which is of the form $e\ v_1 \ldots v_n$.
The \emph{empty} list is returned when no witness can be found after
exploring $k$ traces.

\paragraph{Modeling Semantics}
We resolve the non-determinism in the operational semantics
(\S~\ref{sec:semantics}) via the procedure
$$
\evalN : e \rightarrow \triple{v \cup \stuck}{\vsu}{\tsu}^{*}
$$
Due to the non-determinism introduced by \gensym, a call
$\evalfn{e}$ returns a \emph{list}
of possible results of the form $\triple{v \cup \stuck}{\vsu}{\tsu}$
such that $\steps{e}{\emptysu}{\emptysu}{v \cup \stuck}{\vsu}{\tsu}$.

\paragraph{Currying}
We address the issue of currying by defining a procedure \genArgs{e},
defined in Figure~\ref{fig:algo-gen-args}, that takes as input an
expression $e$ and produces a \emph{saturated} expression of the form
$\eapp{e}{\ehole_1[\thole_1] \ldots \ehole_n[\thole_n]}$ that
\emph{does not} evaluate to a lambda.
This is achieved with a simple loop that keeps adding holes to the
target application until evaluating the term yields a non-lambda value.
%
%

\begin{figure}[t]
$$
\begin{array}{lclr}
\genArgsN   & : & e \rightarrow e \\
\genArgs{e} & = & \mbox{\textbf{case }} \evalfn{e} \mbox{\textbf{ of}} \\
 \quad \triple{\efun{x}{e}}{\vsu}{\tsu},\ldots &\rightarrow& \genArgs{\eapp{e}{\vhole{\thole}}} & (\ehole, \thole \mbox{ are fresh}) \\
 \quad \_ &\rightarrow& e \\
\end{array}
$$
\caption{Generating a saturated application.}
\label{fig:algo-gen-args}

\end{figure}

\paragraph{Generating Witnesses}
Finally, Figure~\ref{fig:algo-gen-witness} summarizes the overall
implementation of our search for witnesses with the procedure
$\genWitness{k}{e}$, which takes as input a bound $k$ and the
target expression $e$, and returns a list of witness expressions
$\eapp{e}{v_1 \ldots v_n}$ that demonstrate how the input program
gets stuck.
The search proceeds as follows.
\begin{enumerate}
  \item We invoke $\genArgs{e}$ to produce a \emph{saturated}
        application $e_{sat}$.

  \item We take the first $k$ traces returned by $\evalN$
        on the target $e_{sat}$, and

  \item We extract the substitutions corresponding to the
        $\stuck$ traces, and use them to return the list
        of witnesses.
\end{enumerate}
We obtain the following corollary of Theorem~\ref{thm:soundness}:

\begin{cor}[Witness Generation]
\label{thm:generation}
  If $\genWitness{k}{e} = \triple{\eapp{e}{v_1 \ldots v_n}}{\vsu}{\tsu},\ldots$
  then for all types $t_1 \ldots t_n$ there exist values $w_1:t_1\ \ldots\ w_n:t_n$ such that
  $\steps{\eapp{e}{w_1 \ldots w_n}}{\emptysu}{\emptysu}{\stuck}{\vsu'}{\tsu'}$.
\end{cor}
\begin{proof}
  For any function $f$ of multiple arguments, we can define $f'$ as the
  uncurried version of $f$ that takes all of its arguments as a single
  nested pair, and then apply Theorem~\ref{thm:soundness} to $f'$.
\end{proof}

\begin{figure}[t]
\centering
$$
\begin{array}{lclr}
\genWitnessN       & : & \mathsf{Nat} \times e \rightarrow e^{*} & \\
\genWitness{n}{e}  & = & \{ \resolve{e_{sat}}{\vsu} \mid \vsu \in \Sigma \} & \\
\quad \mbox{\textbf{where}} &    & & \\
\quad \quad e_{sat} & = & \genArgs{e} & (1) \\
\quad \quad res    & = & \takefn{n}{\evalfn{e_{sat}}} & (2) \\
\quad \quad \Sigma & = & \{ \vsu\ \mid \triple{\stuck}{\vsu}{\tsu} \in res\} & (3)
\end{array}
$$
\caption{Generating witnesses.}
\label{fig:algo-gen-witness}
\end{figure}

\section{Explaining Type Errors With Traces}
\label{sec:interactive}

A trace, on its own, is too detailed to be
a good explanation of the type error. One approach
is to use the witness input to step through the
program with a \emph{debugger} to observe how
the program evolves.
This route is problematic for two reasons.
First, existing debuggers and interpreters for
typed languages (\eg\ \ocaml) typically require
a type-correct program as input.
Second, we wish to have a quicker way to get
to the essence of the error, \eg\ by skipping
over irrelevant sub-computations, and focusing
on the important ones.

In this section we present an interactive visualization of program
executions.
First, we extend our semantics (\S~\ref{sec:inter-semant}) to record
each reduction step in a \emph{trace}, producing a \emph{reduction
  graph} alongside the witness.
Then we describe a set of common \emph{interactive debugging} steps that
can be expressed as simple traversals over the reduction graph
(\S~\ref{sec:traversing-graph}), yielding an interactive debugger that
allows the user to visualize \emph{how} the program goes
(wrong).

\subsection{Tracing Semantics}
\label{sec:inter-semant}

\paragraph{Reduction Graphs}
A \emph{steps-to} edge is a pair of expressions \singlestep{e_1}{e_2}, which
indicates that $e_1$ reduces, in a single step, to $e_2$.
%
%
A \emph{reduction graph} is a set of steps-to edges:
$$\tr ::= \bullet \spmid \singlestep{e}{e}; \tr$$  

\paragraph{Tracing Semantics}
We extend the transition relation (\S~\ref{sec:semantics}) to
collect the set of edges corresponding to the reduction graph.
Concretely, we extend the operational semantics to
a relation of the form $\stepg{e}{\vsu}{\tsu}{\tr}{e'}{\vsu'}{\tsu'}{\tr'}$
where $\tr'$ collects the transitions.

\paragraph{Collecting Edges}
%
%
The general recipe for collecting steps-to edges is
to record the consequent of each original rule in the
trace. That is, each original judgment \step{e}{\vsu}{\tsu}{e'}{\vsu'}{\tsu'}
becomes \stepg{e}{\vsu}{\tsu}{\tr}{e'}{\vsu'}{\tsu'}{\singlestep{e}{e'}; \tr}.
%


\subsection{Interactive Debugging}
\label{sec:traversing-graph}

Next, we show how to build a visual interactive debugger
from the traced semantics, by describing the visualization
\emph{state} --- \ie\ what the user sees at any given moment ---
and the set of \emph{commands} available to user and what
they do.

\paragraph{Visualization State}
A \emph{visualization state} 
is a \emph{directed graph}
whose vertices are expressions and whose edges are such
that each vertex has at most one predecessor and at most one
successor. In other words, the visualization state looks
like a set of linear lists of expressions as shown in
Figure~\ref{fig:nanomaly-factorial}.
The \emph{initial state} is the graph containing a single
edge linking the initial and final expressions.

\begin{figure}[t]
\centering
\includegraphics[height=4in]{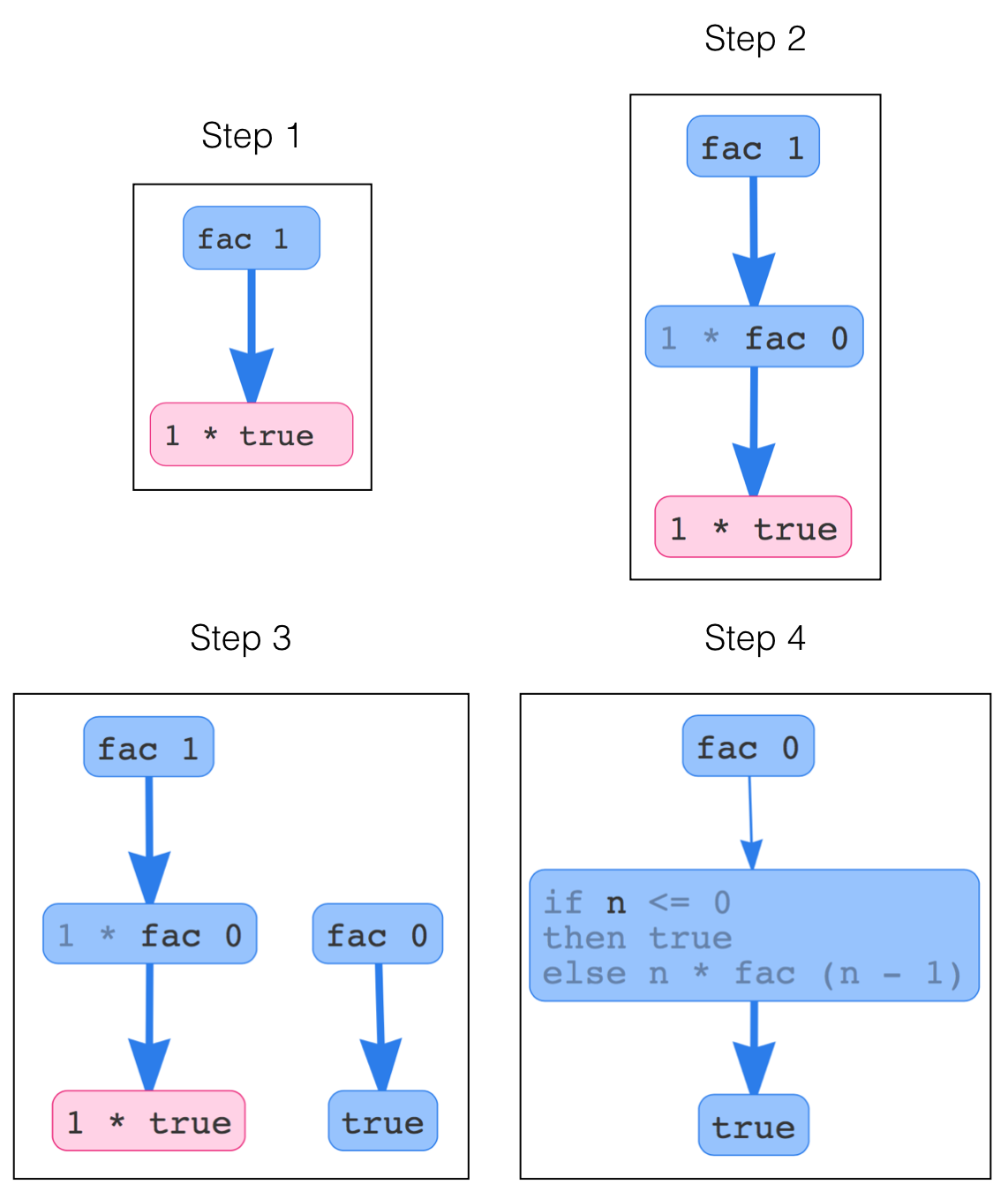}
\caption{A sequence of interactions with the trace of
  \texttt{fac 1}. The stuck term is red, in each node the redex is
  highlighted. Thick arrows denote a multi-step transition, thin arrows
  denote a single-step transition. We start in step 1. In step 2 we jump
  forward from the witness to the next function call. In step 3 we step
  into the recursive \texttt{fac 0} call, which spawns a new ``thread''
  of execution. In step 4 we take a single step forward from
  \texttt{fac 0}.}
\label{fig:nanomaly-factorial}
\end{figure}


\paragraph{Commands}
Our debugger supports the following \emph{commands}, each of which
is parameterized by a single expression (vertex) selected from the
(current) visualization state:
\begin{itemize}
\item \stepforwardsym, \stepbackwardsym:
      show the result of a single step forward or backward;
\item \jumpforwardsym, \jumpbackwardsym:
      show the result of taking multiple steps (a \emph{``big''} step)
      up to the first function call, or return, forward or backward
      respectively;
\item \stepintosym:
      show the result of stepping into a function call in a sub-term,
      isolating it in a new reduction thread; and
\item \stepoversym:
      show the result of skipping over a function call in a sub-term.
\end{itemize}

\paragraph{Jump Compression}
A \emph{jump compressed} trace is one whose edges are limited to forward
or backward jumps.
%
In our experience, jump compression abstracts many details of the
computation that are often uninteresting or irrelevant to the
explanation.
In particular, jump compressed traces hide low-level operations and
summarize function calls as call-return pairs, see
Figure~\ref{fig:fac-jump} for a variant of @fac@ that implements the
subtraction as a function call instead of a primitive.
\begin{figure}[t]
\centering
\includegraphics[height=2in]{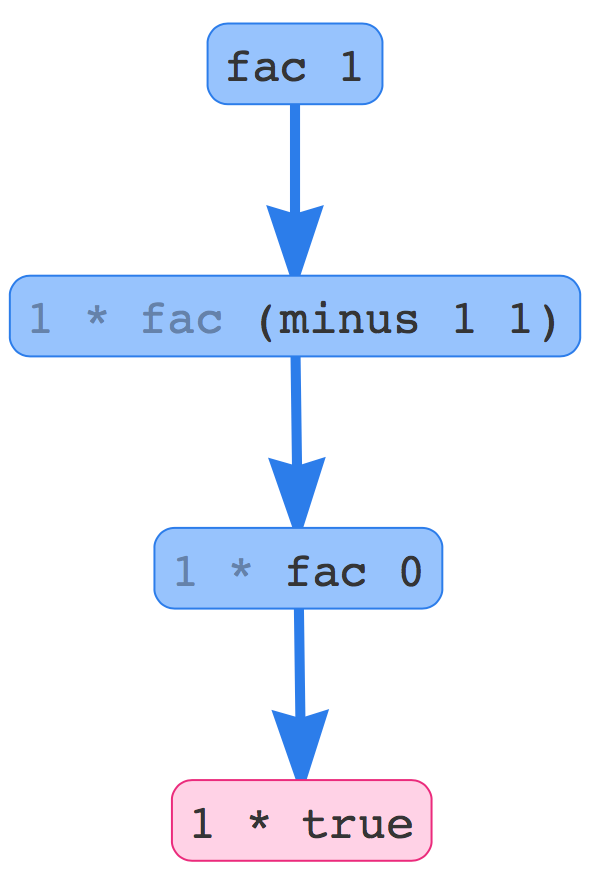}
\caption{Jump-compressed trace of \texttt{fac 1} with subtraction
  implemented as a function call.}
\label{fig:fac-jump}
\end{figure}
Once users have identified interesting call-return pairs, they can
step into those calls and proceed with more fine-grained steps.
%
%
Note that jump compressed traces are not quite the same as
stack traces as they show \emph{all} function calls, including
those that returned successfully.
\section{Evaluation}
\label{sec:evaluation}

We have implemented a prototype of our search procedure and trace
visualization for a purely functional subset of \ocaml\ --- with
polymorphic types and records, but no modules, objects, or polymorphic
variants --- in a tool called \nanomaly.
\footnote{\url{https://github.com/ucsd-progsys/nanomaly}}
We treat explicit type signatures, \eg @(x : int)@, as
primitive operations that narrow the type of the wrapped value.
In our implementation we instantiated \gensym\ with a simple random
generation of values, which we will show suffices for the majority of
type errors.

\paragraph{Evaluation Goals}
There are four questions we seek to answer with our evaluation:
\begin{enumerate}
\item \emphbf{Witness Coverage} (\S~\ref{sec:eval:witness-coverage},~\ref{sec:how-safe})
      How many ill-typed programs \emph{admit} witnesses?
\item \emphbf{Witness Complexity} (\S~\ref{sec:trace-complexity})
      How \emph{complex} are the traces produced by the witnesses?
\item \emphbf{Witness Utility} (\S~\ref{sec:advantage-traces},~\ref{sec:user-study})
      How \emph{helpful} 
      are the witnesses in debugging type errors?
\item \emphbf{Witness-based Blame} (\S~\ref{sec:locating})
      Can witnesses be used to \emph{locate} the source
      of an error?
\end{enumerate}

In the sequel we present our experimental methodology (\S~\ref{sec:methodology})
and then answer the above questions.
However, for the impatient reader, we first summarize our main results:

\paragraph{1. Most Type Errors Admit Witnesses}
Our prime result is that the vast majority of static type errors, around
85\%, do in fact admit a dynamic witness.
Further, \toolname efficiently synthesizes witnesses with its randomized search;
it can synthesize a witness for over 75\% of programs in under one second, \ie
fast enough for interactive use. 

\paragraph{2. Jump-Compressed Traces Are Small}
We find that our jump-compression heuristic effectively abstracts the
pedestrian details of computation, compressing the median trace with
14--15 single-step reductions to only 4 jumps.
Over 80\% of programs have a jump-compressed trace with at most 10
jumps, providing a bird's-eye view from which we can launch a more
in-depth investigation.

\paragraph{3. Witnesses Help Novices}
A witness should also help programmers \emph{understand} and
\emph{fix} type errors.
We use a set of ill-typed student programs to show that \toolname's
witnesses effectively demonstrate the runtime error that the type
system prevented.
Furthermore, we find, in a study of undergraduate students, that
\toolname's witnesses lead to more accurate diagnoses and fixes of type
errors than \ocaml's type error messages.

\paragraph{4. Witnesses Assign Blame}
Finally, we present a simple heuristic that allows us to use witnesses
to \emph{automatically} assign blame for type errors.
We treat the values inside the stuck term as \emph{sources} of typing
constraints and the stuck term itself as a \emph{sink}, producing
a slice of the program that likely contains the error.
Using this heuristic, \toolname's witnesses are competitive with the
state-of-the-art localization tools \mycroft~\cite{Loncaric2016-uk}
and \sherrloc~\cite{Zhang2014-lv}.

\subsection{Methodology}
\label{sec:methodology}
We answer the first two questions on two sets of ill-typed programs,
\ie\ programs that were rejected by the \ocaml\ compiler because of a
type error.
The first dataset comes from the Spring 2014 undergraduate Programming
Languages (CSE 130) course at UC San Diego.
We recorded each interaction with the \ocaml\ top-level system over the
course of the first three assignments (IRB
\#140608),
from which we extracted \ucsdsize\ distinct, ill-typed \ocaml\ programs
from a cohort of 46 students.
The second dataset --- widely used in the literature --- comes from a
graduate-level course at the University of Washington~\cite{Lerner2006-pj},
from which we extracted 284 ill-typed programs.
Both datasets contain relatively small programs, the largest being 348
SLoC; however, they demonstrate a variety of functional programming
idioms including (tail) recursive functions, higher-order functions,
and polymorphic and algebraic data types. 

We answer the third question in two steps.
First, we present a qualitative evaluation of \toolname's traces on a
selection of programs drawn from the UCSD dataset.
Second, we present a quantitative user study of students in the
University of Virginia's Spring 2016 undergraduate Programming Languages
(CS 4501) course.
As part of an exam, we presented the students with ill-typed \ocaml\
programs and asked them to
(1) \emph{explain} the type error, and
(2) \emph{fix} the type error (IRB \#2014009900).
For each problem the students were given the ill-typed program and
either \ocaml's error message or \toolname's jump-compressed trace.

We answer the last question on a subset of the \ucsdbench dataset.
%
For each ill-typed program compiled by a student, we identify the student's
\emph{fix} by searching for the first type-correct program that the student
subsequently compiled.
We then use an expression-level \emph{diff}~\cite{Lempsink2009-xf} to
determine which sub-expressions changed between the ill-typed program
and the student's fix, and treat those expressions as the source of the
type error.

\subsection{Witness Coverage}
\label{sec:eval:witness-coverage}
We ran our search algorithm on each program for 1,000 iterations, with
the entry point set to the function that \ocaml\ had identified as
containing a type error.
Due to the possibility of non-termination we set a timeout of one minute
total per program.
%
%
We also added a na{\"\i}ve check for infinite recursion; at each recursive
function call we check whether the new arguments are identical to the
current arguments.
If so, the function cannot possibly terminate and we report an error.
While not a \emph{type error}, infinite recursion is still a clear bug
in the program, and thus valuable feedback for the user.

\begin{figure}[t]
\centering
\includegraphics[width=0.7\linewidth]{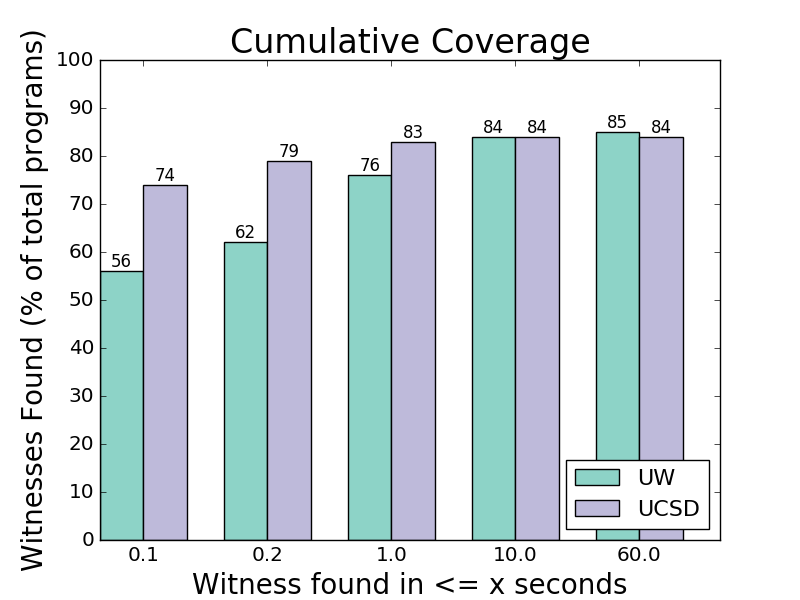}
\caption{Results of our coverage testing. Our random search successfully
  finds witnesses for 76--83\% of the programs in under one second,
  improving to 84--85\% in under 10 seconds. }
\label{fig:results-witness}
\end{figure}
\begin{figure}[t]
\includegraphics[width=0.7\linewidth]{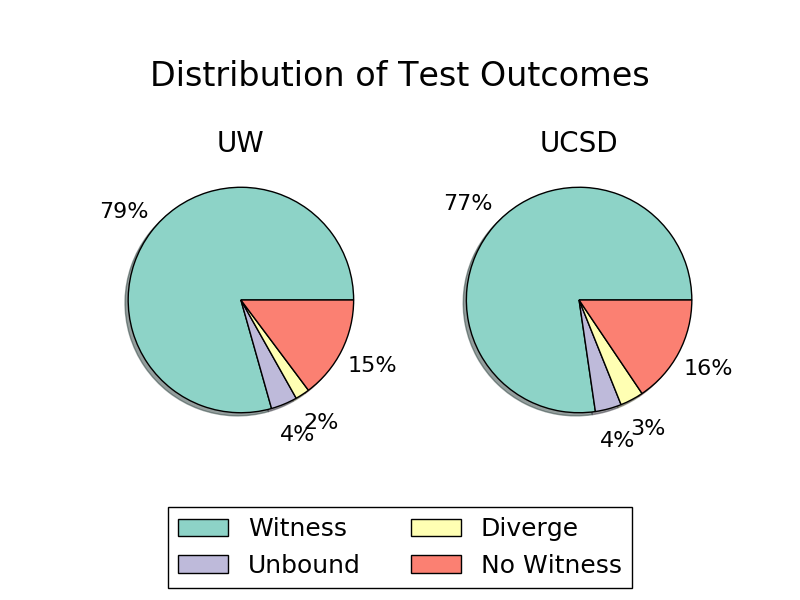}
\caption{Distribution of test outcomes. In both datasets we detect
  actual type errors at least 77\% of the time, unbound variables or
  constructors 4\% of the time, and diverging loops 2--3\% of the
  time. For the remaining 15--16\% of the programs we are unable to
  provide any useful feedback. }
\label{fig:results-distrib}
\end{figure}

\paragraph{Results}
\label{sec:results-witness}
The results of our experiments are summarized in
Figures~\ref{fig:results-witness}~and~\ref{fig:results-distrib}.
In both datasets our tool was able to find a witness for over 75\% of the
programs in under one second, \ie\ fast enough to be integrated as a
compile-time check. If we extend our tolerance to a 10 second timeout,
we reach 84\% coverage, and if we allow a 60 second search,
we hit a maximum of 84--85\% coverage.
Interestingly, while the vast majority of witnesses corresponded to a
type-error, as expected, 4\% triggered an unbound variable error (even
though \ocaml\ reported a type error) and 3\% triggered an infinite
recursion error.
For the remaining 15--16\% of programs we were unable to provide any useful
feedback as they either completed 1,000 tests successfully, or timed out
after one minute.
%
%
While a more advanced search procedure, \eg\ dynamic-symbolic execution,
could likely uncover more errors, our experiments suggest that
type errors are coarse enough (or that novice programs are \emph{simple}
enough) that these techniques are not necessary.

\subsection{How safe are the ``safe'' programs?}
\label{sec:how-safe}

An immediate question arises regarding the 15--16\% of programs for
which we could not synthesize a witness:
are they \emph{actually} safe (\ie is the type system being too conservative),
or did \toolname simply fail to find a witness?
%

To answer this question, we investigated the 732 \ucsdbench programs for
which we failed to find a witness.
We used a combination of automatic and manual coding to categorize these
programs into four classes.
The first class is easily detected by \toolname itself, and thus admits
a precise count.
This left us with 504 programs that required manual coding; we selected
a random sample of 50 programs to investigate, and will report results
based on that sample.
Figure~\ref{fig:no-witness} summarizes the results of our investigation ---
we note the classes that were based on the random sample with a ``*''.
Note that the percentages referenced in Figure~\ref{fig:no-witness} (and
in the sequel) are with respect to the total number of programs in the
\ucsdbench dataset, not only those were \toolname failed to find a
witness.

\begin{figure}[t]
\includegraphics[width=0.7\linewidth]{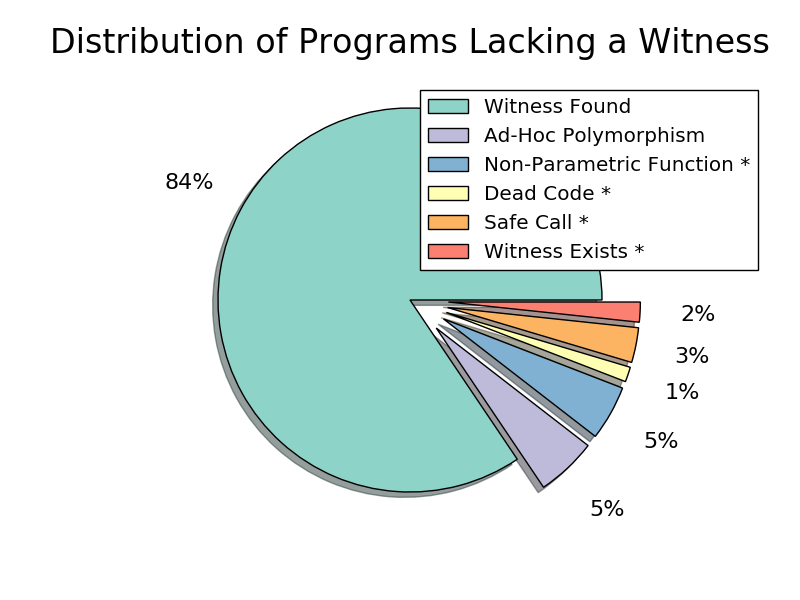}
\vspace{-0.75cm}
\caption{Results of our investigation into programs where \toolname
  did not produce a witness. A ``*'' denotes that the percentage is an
  estimate based on a random sampling of 50 programs.}
\label{fig:no-witness}
\end{figure}

\paragraph{Ad-hoc Polymorphism}
We found that for 5\% of all programs \toolname got stuck when it tried to
compare two holes.
\ocaml provides polymorphic equality and comparison operators,
overloading them for each type.
While convenient to use, they pose a challenge for \toolname's
combination of execution and inference.
For example, consider the \hbox{@n <= 0@} test in our @fac@ example.
The @<=@ operator is polymorphic, but in this case we can make progress because
the literal @0@ is not.
Suppose, however, we parameterized @fac@ by a lower bound, \eg
\begin{code}
  let rec fac n m =
    if n <= m then
      true
    else
      n * fac (n - 1) m
\end{code}
When given @fac@, \toolname will generate two fresh holes
$\nu_1[\alpha_1]$ and $\nu_2[\alpha_2]$ and proceed directly into the
@n < m@ comparison.
We cannot (yet) instantiate either hole because we have no constraints
on the $\alpha$s (we know they must be equal, but nothing else), and
furthermore we do not know what constraints we may encounter later on in
the program.
Thus, we cannot perform the comparison and proceed, and must give up our
search for a witness, even though one obviously exists, any pair of |n|
and |m| such that |n <= m| is false.

Extending \toolname with support for symbolic execution would alleviate
this issue, as we could then begin symbolically executing the program
until we learn how to instantiate |n| and |m|.
Alternatively, we could \emph{speculatively} instantiate both |n| and
|m| with some arbitrary type, and proceed with execution until we
discover a type error.
This speculative instantiation is, of course, unsound; we would have to
take care to avoid reporting frivolous type errors that were caused by
such instantiations.
We would need to track which holes were instantiated speculatively 
to distinguish type errors that would have happened regardless, as
in @fac@, from type errors that were caused by our instantiation.

\newpage
Further, suppose that our speculative instantiation induces a frivolous
type error.
For example, suppose we are given
\begin{code}
  let bad x y =
    if x < y then
      x *. y
    else
      0.0
\end{code}
and choose to (speculatively) instantiate @x@ and @y@ as @int@s and proceed
down the ``true'' branch.
We will quickly discover this was the wrong choice, as they are immediately
narrowed to @float@s.
We must now backtrack and try a different instantiation, but we no
longer need to choose one at random.
Since our instantiation was speculative, and @x@ and @y@
were \emph{originally} holes, we can treat the @*.@ operator as a normal
narrowing point with two holes.
This tells us that the \emph{correct} instantiation was in fact @float@,
and we can then proceed as normal from the backtracking point with a
concrete choice of @float@s.
Thus, it appears that speculative instantiation of holes may be a
useful, lightweight alternative to symbolic execution for our purposes.


\paragraph{Non-Parametric Function Type *}
5\% of all programs lack a witness in our semantics due to our
non-parametric $\tfun$ type for functions.
Recall that our goal is to expose the runtime errors that would have
been prevented by the type systems.
At runtime, it is always safe to call a function, thus we give functions
a simple type $\tfun$ that says they may be applied, but says nothing
about their inputs or outputs.
But consider the following @clone@ function, which is supposed to
produce a list containing @n@ copies of the input @x@.
%
%
\begin{code}
  let rec clone x n =
    if n > 0 then
      clone [x] (n - 1)
    else
      []
\end{code}
Unfortunately, the student instead constructs an @n@-level nested list
containing a single @x@.
The \ocaml compiler rejects this program because the recursive call to
@clone@ induces a cyclic typing constraint @'a = 'a list@, capturing the
fact that each call increases the nesting of the list.
\toolname fails to catch this because we do not track the types of the
inputs to @clone@.

We note, however, that @clone@ cannot go wrong; it is perfectly safe to
repeatedly enclose a list inside another (disregarding the fact that the
nested list is never returned).
Still, such a function would be very difficult to \emph{call} safely, as
the programmer would have to reason about the dependency between the
input @n@ and the nesting of the output list, which cannot be expressed
in \ocaml's type system.

Thus, it is not particularly satisfying that \toolname fails to produce
a witness here; one solution could be to track the types of the
inputs, and demonstrate to the user how they change between recursive
calls.
This would require maintaining a typing environment of variables in
addition to the environments we maintain for holes.
We would have to modify the rule $\reappgood$ from
Figure~\ref{fig:operational} to additionally $\forcesym$ the function's
type against the concrete inputs.
However, we would want to ensure that this $\forcesym$ cannot fail ---
it is preferable to report a stuck term as that provides a fuller view
of the error.
Rather, we would note which evaluation steps induced incompatible
type refinements, and if a traditional witness cannot be found, we could
then report a trace expanded to show precisely these steps.
This represents only a modest extension to our semantics and would be
interesting to explore further.

The $\tfun$ abstraction is also problematic when we have to generate new
functions. Consider the following @pipe@ function that composes a list
of functions.
\begin{code}
  let pipe fs =
    let f a x = a x in
    let base a = a in
    List.fold_left f base fs
\end{code}
In the folding function @f@, the student has applied the accumulator @a@
to the new function @x@, rather than composing the two.
The \ocaml compiler again detects a cyclic typing constraint and rejects
the program, but \toolname is unable to produce a witness.
In this case the issue lies in the fact that the safety of a call to
@pipe@ is determined by its arguments: the call @pipe [(fun x -> 1)]@ is
safe, but the call @pipe [(fun x -> 1); (fun x -> 2); (fun x -> 3)]@
will get stuck when we try to reduce @1 (fun x -> 3)@.
Unfortunately, \toolname is unable to synthesize such a witness because
of the abstraction to $\tfun$.
Specifically, this abstraction forces our hand inside $\gensym$; we do
not know the what the input and output types of the function should be,
so the only safe thing to do is to generate a function that accepts any
input and returns a value of a yet-to-be-determined type.
Thus, our lenient instantiation of holes prevents us from
discovering a witness here.

\paragraph{Dead Code and ``Safe'' Function Calls *}
4\% of all programs contained type errors that were unreachable, either
because they were dead code, or because the student called the function
with inputs that could not trigger the error.

1\% contained type errors that were unreachable by any inputs, often due
to overlapping patterns in a @match@ expression.
%
%
While technically safe, dead code is generally considered a maintenance
risk, as the programmer may not realize that it is dead~\cite{Wheeler2014-fg}
or may accidentally bring it back to life~\cite{Seven2014-gf}.
Thus, a warning like that provided by \ocaml's pattern exhaustiveness
checker would be helpful.

A further 3\% included a function call where the
student supplied ill-typed inputs, but the path induced
by the call did not contain an error.
Consider the following |assoc| function, which
looks up a key in an \emph{association list}, returning a default if
it cannot be found.
\begin{code}
  let rec assoc (d, k, l) = match l with
    | (ki, vi)::tl ->
       if ki = k then
         vi
       else
         assoc (d, k, tl)
    | _ -> d

  let _ = assoc ([], 123, [(123, "sad"); (321, "happy")])
\end{code}
The student's definition of |assoc| is correct, but \ocaml rejects their
subsequent call because the default value |[]| is incompatible with the
|string| values in the list.
In this particular call the key |123| is in the list, so the default
will not be used (even if it were, there would not be an error) and
\ocaml's complaint is moot.
Of course, \ocaml cannot be expected to know that this particular call
is safe, its type system is not sophisticated enough to express the
necessary conditions.
%


\paragraph{Witness Exists *}
We found that only 2\% of all programs admit a witness that \toolname
was unable to discover.
Slightly over half involved synthesizing a \emph{pair} of
specially-crafted inputs that would result in the function returning
values of incompatible types.
The rest required synthesizing an input that would trigger a
particular path through the program, and would likely have been caught
by symbolic execution.

\paragraph{Summary}
Our investigation suggests that the majority of programs
for which we fail to find a witness do not, in fact, admit a
witness under \toolname's semantics.
%
These programs were generally cases where \ocaml's type system was
overly conservative.
Of course, the conservatism is somewhat justified as each case pointed
to code that would be difficult to use or maintain;
it would be interesting to investigate how demonstrate these issues in an
intuitive manner.



\subsection{Witness Complexity}
\label{sec:trace-complexity}

For each of the ill-typed programs for which we could
find a witness, we measure the complexity of the generated
trace using two metrics.

%
\begin{enumerate}
\item \emphbf{Single-step:} The size of the trace after expanding
  all of the single-step edges from the witness to the stuck term, and
\item \emphbf{Jump-compressed:} The size of the jump-compressed trace.
\end{enumerate}

%
\begin{figure}[t]
\includegraphics[width=0.7\linewidth]{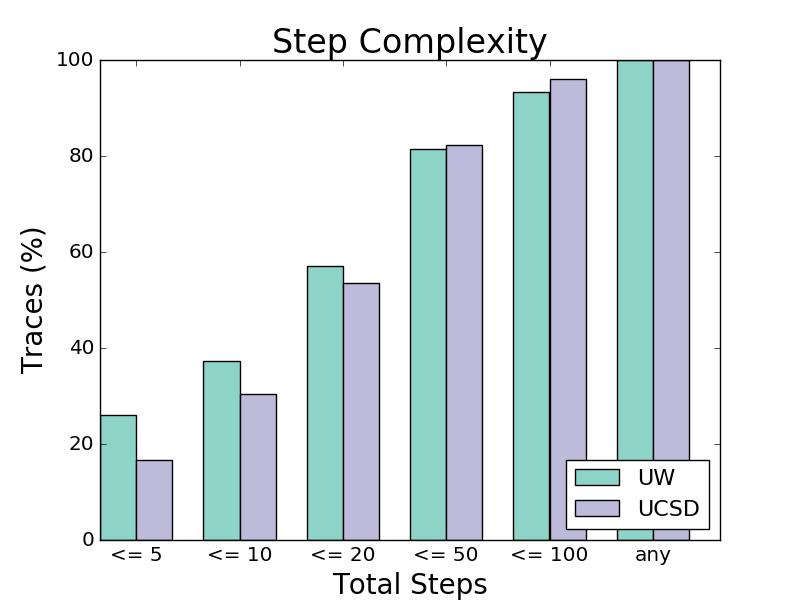}
\includegraphics[width=0.7\linewidth]{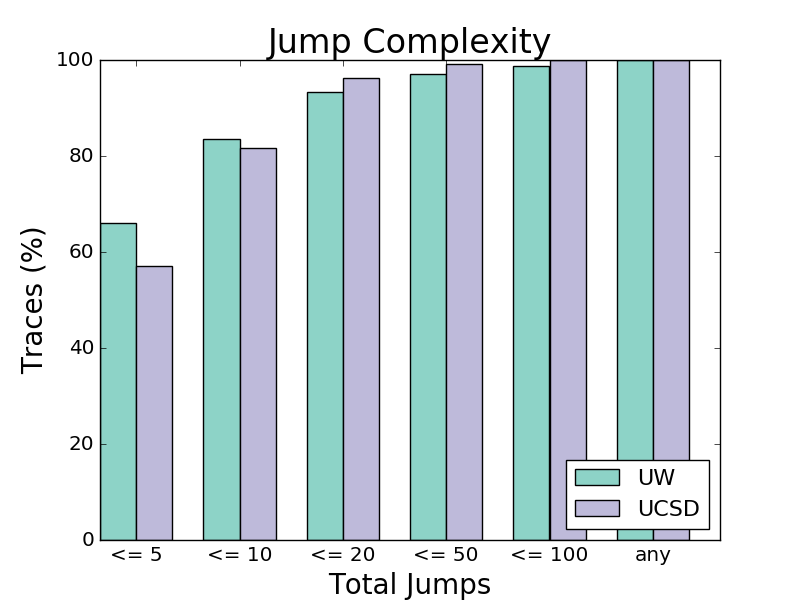}
\caption{Complexity of the generated traces. Over 80\% of the combined traces
  have a jump complexity of at most 10, with an average complexity of 7
  and a median of 5.}
\label{fig:results-complexity}
\end{figure}
\paragraph{Results}
\label{sec:results-complexity}
The results of the experiment are summarized in
Figure~\ref{fig:results-complexity}.
The average number of single-step reductions per trace is 17 for the
\ucsdbench\ dataset (42 for the \uwbench\ dataset) with a maximum of
2,745 (\resp 982) and a median of 15 (\resp 15).
The average number of jumps per trace is 7 (\resp 9) with a
maximium of 353 (\resp 221) and a median of 4 (\resp 4).
In both datasets about 60\% of traces have at most 5 jumps, and 80\% or more
have at most 10 jumps.

\subsection{Qualitative Evaluation of Witness Utility}\label{sec:advantage-traces}

Next, we present a \emph{qualitative} evaluation that compares
the explanations provided by \toolname's dynamic witnesses with
the static reports produced by the \ocaml\ compiler and \sherrloc,
a state-of-the-art fault localization approach~\cite{Zhang2014-lv}.
In particular, we illustrate, using a series of examples drawn
from student programs in the \ucsdbench\ dataset, how \toolname's
jump-compressed traces can get to the heart of the error. Our approach
highlights the conflicting values that cause the program to get
stuck, rather that blaming a single one,
shows the steps necessary to reach the stuck state, and
does not assume that a function is correct just because it type-checks.
For each example we will present:
(1)~the code;
(2)~the error message returned \ocaml;
(3)~the error locations returned by \hlOcaml{\ocaml} and \hlSherrloc{\sherrloc};
and (4)~\toolname's jump-compressed trace.



\paragraph{Example: Recursion with Bad Operator}
The recursive function @sqsum@ should square each
element of the input list and then compute the sum
of the result.
\begin{ecode}
  let rec sqsum xs = match xs with
    | [] -> 0
    | h::t -> (*@\hlOcaml{\hlSherrloc{sqsum t}}@*) @ (h * h)
\end{ecode}
Unfortunately the student has used the list-append
operator |@| instead of \texttt{+}. 
Both \ocaml\ and \sherrloc\ blame the \emph{wrong location},
the recursive call @sqsum t@, with the message
\begin{verbatim}
  This expression has type
    int
  but an expression was expected of type
    'a list
\end{verbatim}
\toolname\ produces a trace showing how the evaluation of
@sqsum [1]@ gets stuck.
\begin{center}
  \includegraphics[height=125px]{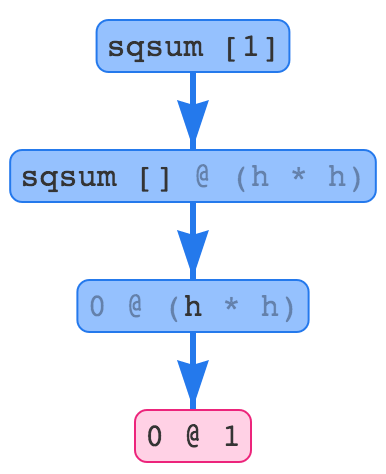}
\end{center}
The trace highlights the entire stuck term
(not just the recursive call), emphasizing
the \emph{conflict} between @int@ and @list@
rather than assuming one or the other is correct.

\paragraph{Example: Recursion with Bad Base Case}
The function @sumList@ should add up
the elements of its input list.
\begin{ecode}
  let rec sumList xs = match xs with
    | []    -> (*@\hlSherrloc{[]}@*)
    | y::ys -> y + (*@\hlOcaml{sumList ys}@*)
\end{ecode}
Unfortunately, in the base case, it returns @[]@
instead of @0@.
\sherrloc\ blames the base case, and \ocaml\
assumes the base case is correct and blames
the \emph{recursive call} on line 3:
\begin{verbatim}
  This expression has type
    'a list
  but an expression was expected of type
    int
\end{verbatim}
Both of the above are parts of the full story, which
is summarized by \toolname's trace showing
how @sumList [1; 2]@ gets stuck at @2 + []@.
\begin{center}
  \includegraphics[height=125px]{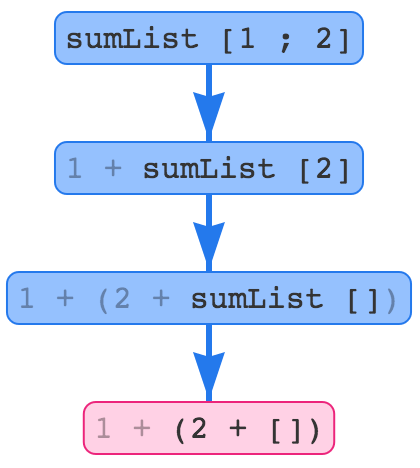}
\end{center}
The trace clarifies (via the third step)
that the @[]@ results from the recursive call
\hbox{@sumList []@,} and shows how it is incompatible with
the subsequent \texttt{+} operation.


\paragraph{Example: Bad Helper Function that Type-Checks}
The function @digitsOfInt@ should return a list of
the digits of the input integer.
\begin{ecode}
  let append x xs =
    match xs with
    | [] -> [x]
    | _  -> x :: xs

  let rec digitsOfInt n =
    if n <= 0 then
      []
    else
      append ((*@\hlSherrloc{digitsOfInt (n / 10)}@*)) [(*@\hlOcaml{n mod 10}@*)]
\end{ecode}
%
Unfortunately, the student's @append@ function \emph{conses} an element
onto a list instead of appending two lists.
Though incorrect, @append@ still type-checks and thus \ocaml and
\sherrloc blame the \emph{use-site} on line 10.
\begin{verbatim}
  This expression has type
    int
  but an expression was expected of type
    'a list
\end{verbatim}
In contrast, \toolname makes no assumptions about @append@,
yielding a trace that illustrates the error on line 4, by
highlighting the conflict in consing a list onto a list of integers.
\begin{center}
  \includegraphics[height=160px]{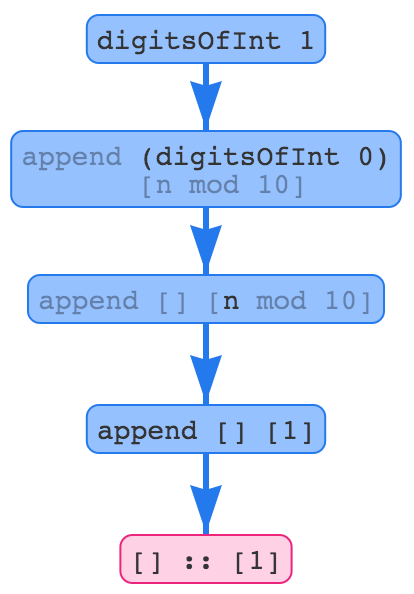}
\end{center}

\paragraph{Example: Higher-Order Functions}
The higher-order function @wwhile@ is supposed
to emulate a traditional while-loop. It takes
a function @f@ and repeatedly calls @f@ on the
first element of its output pair, starting with
the initial @b@, till the second element is @false@.
\newpage
\begin{ecode}
  let rec wwhile (f,b) =
    match f with
    | (z, false) -> z
    | (z, true)  -> wwhile (f, z)

  let f x =
    let xx = x * x in
    (xx, (xx < 100))

  let _ = wwhile ((*@\hlOcaml{\hlSherrloc{f}}@*), 2)
\end{ecode}
The student has forgotten to \emph{apply} @f@ at all on line 2,
and just matches it directly against a pair.
This faulty @wwhile@ definition nevertheless typechecks,
and is assumed to be correct by both \ocaml\ and \sherrloc\
which blame the use-site on line 10.
\begin{verbatim}
  This expression has type
    int -> int * bool
  but an expression was expected of type
    'a * bool
\end{verbatim}
\toolname\ synthesizes a trace that draws the eye to the
true error: the @match@ expression on line 2, and highlights
the conflict in matching a function against a pair pattern.
\begin{center}
  \includegraphics[height=135px]{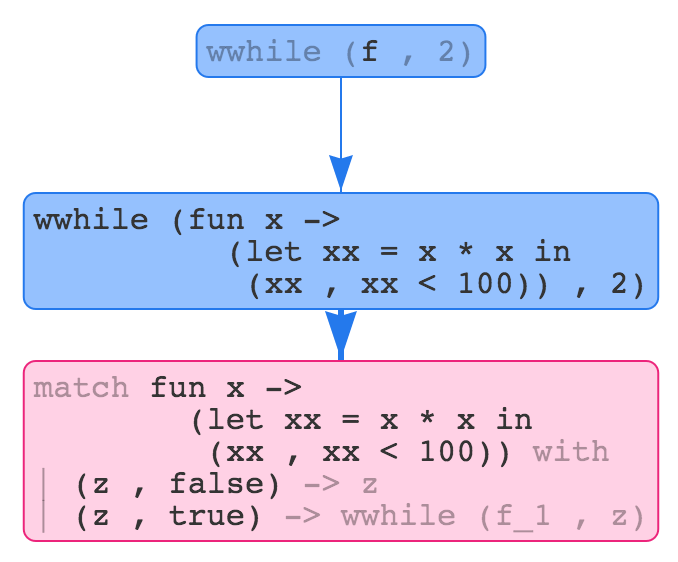}
\end{center}
By highlighting conflicting values, \ie\ the source and sink
of the problem, and not making assumptions about function correctness, \toolname\
focusses the user's attention on the piece of code that is
actually relevant to the error.

\subsection{Quantitative Evaluation of Witness Utility}
\label{sec:user-study}
%
We assigned four problems to the 60 students in the course: the
@sumList@, \hbox{@digitsOfInt@,} and @wwhile@ programs from
\S~\ref{sec:advantage-traces}, as well as the following @append@ program
\begin{ecode}
  let append x l =
    match x with
    | []   -> l
    | h::t -> h :: t :: l
\end{ecode}
which triggers an occurs-check error on line 4.
For each problem the students were given the ill-typed program and
either \ocaml's error or \toolname's jump-compressed trace;
the full user study is available in Appendix~\ref{sec:user-study-exams}.
Due to the nature of an in-class exam, not every student answered every
question; we received between 13 and 28 (out of a possible 30) responses
for each problem-tool pair.

We then instructed four annotators (one of whom is an author, the other
three are teaching assistants at UCSD) to classify the answers as
correct or incorrect.
We performed an inter-rater reliability (IRR) analysis to determine the
degree to which the annotators consistently graded the exams.
\footnote{Measuring IRR is an established practice to account for
  potential bias among raters. The students were asked to explain the
  errors in English; judging whether they truly understood the errors
  involved a surprising amount of subjectivity, and is thus subject to
  rater bias.}
As we had more than two annotators assigning nominal (``correct'' or
``incorrect'') ratings we used Fleiss' kappa~\cite{Fleiss1971-du} to
measure IRR.\@
Fleiss' kappa is measured on a scale from $1$, indicating total
agreement, to $-1$, indicating total disagreement, with $0$ indicating
random agreement.

Finally, we used a one-sided Mann-Whitney $U$ test~\cite{Mann1947-fd} to
determine the significance of our results.
The null hypothesis was that the responses from students given
\toolname's witnesses were drawn from the same distribution as those
given \ocaml's errors, \ie \toolname had no effect.
Since we used a one-sided test, the alternative to the null hypothesis
is that \toolname had a \emph{positive} effect on the responses.
We reject the null hypothesis in favor of the alternative if the test
produces a significance level $p < 0.05$, a standard threshold for
determining statistical significance.

\begin{figure}[t]
\includegraphics[width=0.7\linewidth]{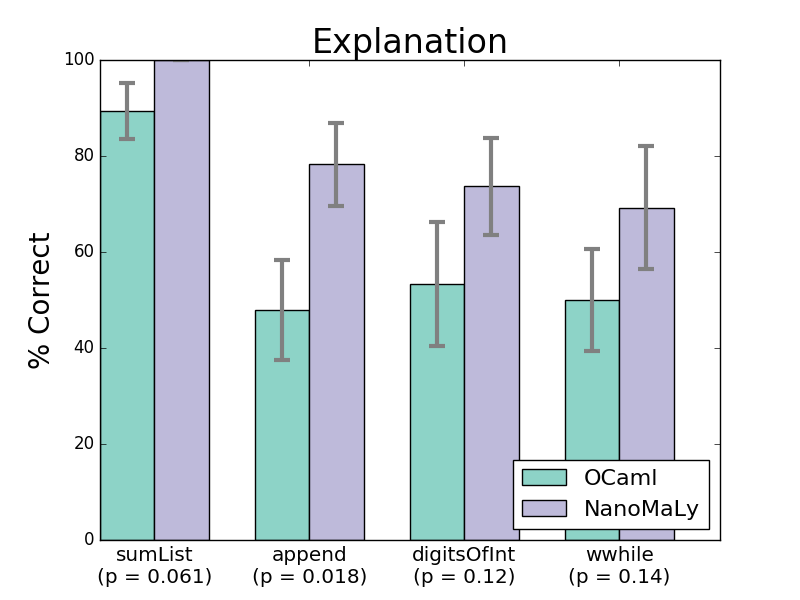}
\includegraphics[width=0.7\linewidth]{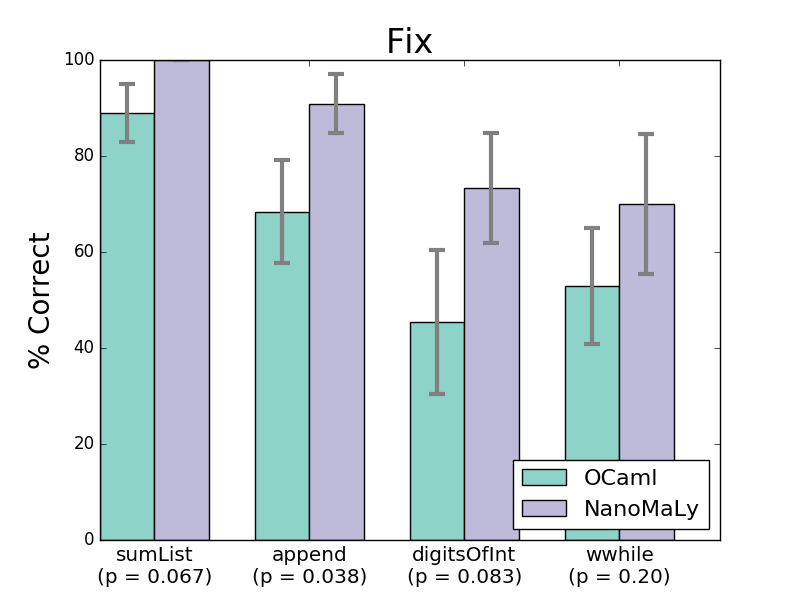}
\caption{A classification of students' explanations and fixes for type
  errors, given either \ocaml's error message or \toolname's
  jump-compressed trace. The students given \toolname's jump-compressed
  trace consistently scored better ($\ge 10\%$) than those given
  \ocaml's type error. We report the result of a one-sided Mann-Whitney
  $U$ test for statistical significance in parentheses.}
\label{fig:results-user-study}
\end{figure}

\paragraph{Results}
The measured kappa values were $\kappa = 0.72$ for the explanations and
$\kappa = 0.83$ for the fixes; while there is no formal notion for what
consititutes strong agreement~\cite{Krippendorff2012-wd}, kappa values
above $0.60$ are often called ``substantial''
agreement~\cite{Landis1977-ey}.
Figure~\ref{fig:results-user-study} summarizes a single annotator's
results, which show that students given \toolname's jump-compressed
trace were consistently more likely to correctly explain
and fix the type error than those given \ocaml's error message.
Across each problem the \toolname responses were marked correct
$10-30\%$ more often than the \ocaml responses, which suggests that
the students who had access to \toolname's traces had a better
understanding of the type errors;
however, only the @append@ tests were statistically significant at
$p < 0.05$.

\paragraph{Threats to Validity}
Measuring understanding is a difficult task; the following summarize
the threats to the validity of our results.

\subparagraph{Construct.}
We used the correctness of the student's explanation of, and fix for,
the type error as a proxy for her understanding, but it is possible that
other metrics would produce different results. One such metric that is
also surely relevant is time-to-completion, \ie a good error report
should \emph{quickly} guide the student to a fix. Unfortunately, the
in-class exam setting of our study did not admit the collection of
timing data.

Furthermore, one might object to our selection of \ocaml as
the baseline comparison rather than \sherrloc or \mycroft, which also
claim to produce more accurate error reports. This is indeed a
limitation of our study, but we note that \sherrloc blames the same
expression as \ocaml in both @wwhile@ and @append@, and in @digitsOfInt@
both blame the wrong function.

Finally, one might point out that our study investigates the use of
\toolname as a \emph{debugging} aid rather than as a \emph{teaching}
aid. That is, it may be that \toolname helps students solve their
immediate problem, but does not help them build a lasting understanding
of the type system. We have not attempted a longitudinal study of the
long-term impact of using \toolname, but we agree that it would be an
interesting future direction.

\subparagraph{Internal.}
We assigned students randomly to two groups. The first was given
\ocaml's errors for @append@ and \hbox{@digitsOfInt@,} and \toolname's trace
for @sumList@ and \hbox{@wwhile@;} the second was given the opposite
assignment of errors and traces. This assignment ensured that: (1) each
student was given \ocaml and \toolname problems; and (2) each student
was given an ``easy'' and ``hard'' problem for both \ocaml and
\toolname. Students without sufficient knowledge of \ocaml could affect
the results, as could the time-constrained nature of an exam. For these
reasons we excluded any answers left blank from our analysis.

\subparagraph{External.}
Our experiment used students in the process of learning \ocaml,
and thus may not generalize to all developers. The four
programs were chosen manually, via a random selection and
filtering of the programs in the \ucsdbench dataset. In some cases we made
minor simplifying edits (\eg alpha-renaming, dead-code removal) to the
programs to make them more understandable in the short timeframe of an
exam; however, we never altered the resulting type-error. A different
selection of programs may lead to different results.

\subparagraph{Conclusion.}
We collected exams from 60 students, though due to the nature of the
study not every student completed every problem.
The number of complete submissions ranges from 13 (for the \toolname
version of @wwhile@) to 28 (for the \ocaml version of @sumList@), out of
a maximum of 30 per program-tool pair.
Our results are statistically significant in only 2 out of 8 tests; however,
collecting more responses per test pair was not possible as it
would require having students answer the same problem twice (once with
\ocaml and once with \toolname).

\subsection{Locating Errors with Witnesses}
\label{sec:locating}

We have seen that \toolname can effectively synthesize witnesses to
explain the majority of (novice) type errors, but a good error report
should also help \emph{locate} the source of the error.
Thus, our final experiment seeks to use \toolname's witnesses as
localizations.

As discussed in \S~\ref{sec:methodology}, we recorded
each interaction of our students with the \ocaml top-level system.
This means that, in addition to collecting ill-typed programs, we
collected subsequent, fixed versions of the same programs.
%
For each ill-typed program compiled by a student, we identify the student's
\emph{fix} by searching for the first type-correct program that the student
subsequently compiled.
We then use an expression-level \emph{diff}~\cite{Lempsink2009-xf} to
determine which sub-expressions changed between the ill-typed program
and the student's fix, and treat those expressions as the source of the
type error.

Not all ill-typed programs will have an associated fix; furthermore,
at some point a ``fix'' becomes a ``rewrite''.
We do not wish to consider the ``rewrites'', so we discard outliers
where the fraction of expressions that have changed is more than one
standard deviation above the mean, establishing a diff threshold of
40\%.
This accounts for roughly 14\% of programs pairs we discovered, leaving
us with 2,710 program pairs.

For each pair of an ill-typed program and its fix, we run \toolname and
collect two sets of source locations:
(1) the source location corresponding to the stuck term; and
(2) the source locations that \emph{produced} the values inside the
stuck term.
Intuitively, these two classes of locations correspond to \emph{sinks}
and \emph{sources} for typing constraints.
For example, in the @sqsum@ program from \S~\ref{sec:advantage-traces}
the stuck term is \verb!0 @ 1!.
This corresponds to the call to \verb!@! on line 3, and contains
the literal @0@ from line 2 and the value @1@ produced by the
@*@ on line 3.

We compare \toolname's witness-based predictions against a baseline of
the \ocaml compiler as well as the state-of-the-art 
localization tools \sherrloc and \mycroft.
\sherrloc~\cite{Zhang2014-lv} attempts to predict the most likely source
of a type error by searching the typing constraint graph for constraints
that participate in many unsatisfiable paths and few satisfiable paths.
\mycroft~\cite{Loncaric2016-uk} reduces the localization problem to
MaxSAT by searching for a minimal subset of constraints that can be
removed, such that the resulting system is satisfiable.
Both tools produce a \emph{set} of equally-likely expressions to blame
for the error (in practice the set contains only a few expressions),
similar to \toolname's witness-based predictions.

We evaluate each tool based on whether \emph{any} of its predictions
identifies a changed expression.
There were a number of programs where \mycroft or \sherrloc
encountered an unsupported language feature or timed out after one
minute, or where \toolname failed to produce a witness.
We discard all such programs in our evaluation to level the playing
field, around 15\% for each tool, leaving us with a benchmark set of
1,759 programs.

\begin{figure}[t]
\includegraphics[width=0.7\linewidth]{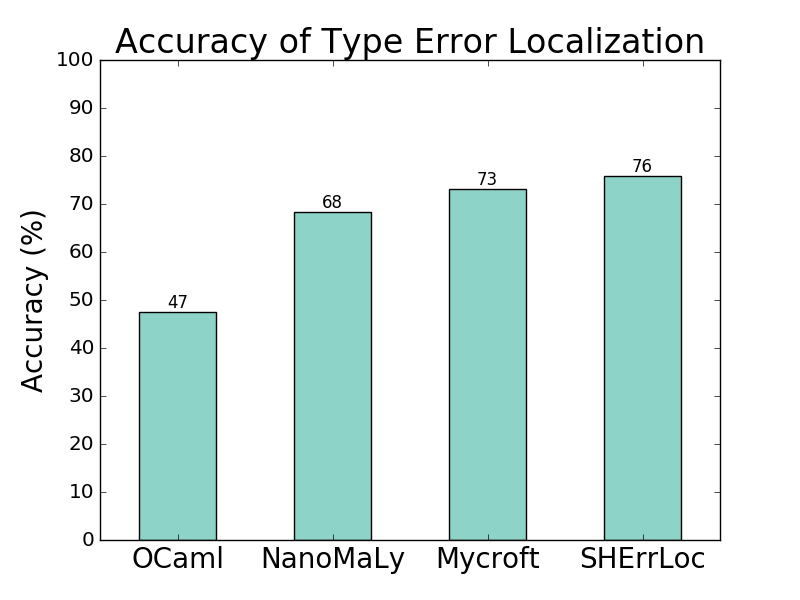}
\caption{Accuracy of type error localization. \toolname's witness-based
  predictions outperform \ocaml by 21 points, and are competitive
  with the state-of-the-art tools \mycroft and \sherrloc.}
\label{fig:results-blame}
\end{figure}

\paragraph{Results}
Figure~\ref{fig:results-blame} summarizes our results, which show that
\toolname's witnesses are competitive with \mycroft and \sherrloc in
automatically locating the source of a type error.
\toolname, \mycroft, and \sherrloc all outperform the \ocaml compiler,
which is not surprising given that they can produce multiple possible
error locations, while the \ocaml compiler is limited to one predicted
error location.
Interestingly, while all tools have a median of 2 predicted error
locations per program, \mycroft and \sherrloc have a long tail with a
maximum of 22 (\resp 11) locations, while \toolname's maximum is 5
locations.
We also note that while \mycroft and \sherrloc were designed
specifically to \emph{localize} type errors, \toolname's foremost
purpose is to \emph{explain} them, we consider its ability to localize
type errors an added benefit.

\paragraph{Threats to Validity}
Our benchmarks were drawn from students in an undergraduate course
at \ucsdbench\ and may not be representative of other student bodies.
We mitigate this threat with a large empirical evaluation of 1,759
programs, drawn from a cohort of 46 students.
A similar threat is that students are not industrial programmers, thus
our results may not translate to large-scale software engineering.
%
%
However, in our experience programmers are able to construct a mental
model of type systems after sufficient exposure, at which point
traditional error reports may suffice.
We are thus particularly interested in aiding novice programmers as
they learn to work with the type system.

Our definition of the next well-typed program as the intended ground
truth answer is another threat to validity. Students might submit
multiple well-typed ``rewrites'' between the initial ill-typed program
and the final intended answer.
Our approach to discarding outliers is intended to mitigate this threat.
A similar threat is our removal of programs where any of the tools could
not produce an answer.
It may be, for example, that \mycroft and \sherrloc are particularly
effective on programs that do not admit dynamic witnesses.
Finally, our use of student fixes as oracles for the source of type
errors assumes that students are able to correctly identify the source.
As the students are in the process of learning \ocaml and the type
system, this assumption may be faulty, \emph{expert} users may disagree
with the student fixes.
We believe, however, that it is reasonable to use student fixes as
oracles, as the student is the best judge of what she \emph{intended} to
do.


\subsection{Discussion}
\label{sec:discussion}

To summarize, our experiments demonstrate that \nanomaly finds witnesses
to type errors:
(1) with high coverage in a timespan amenable to compile-time analysis;
(2) with traces that have a low median complexity of 5 jumps;
(3) that are more helpful to novice programmers than traditional type
error messages; and
(4) that can be used to automatically locate the source of a type error.

There are, of course, drawbacks to our approach.
%
In the sequel we discuss a selection of drawbacks, and how we might
address them in future work.

\paragraph{Random Generation}
Random test generation has difficulty generating highly constrained
values, \eg\ red-black trees or a pair of equal integers. If the type
error is hidden behind a complex branch condition \nanomaly\ may not be
able to trigger it. Exhaustive testing and dynamic-symbolic execution
can address this short-coming by performing an exhaustive search for
inputs (\emph{resp}.\ paths through the program). Our approach does not
rely on random generation, we could easily substitute it for
dynamic-symbolic execution by extending the evaluation relation to
maintain a path condition and replacing the \gensym function with a call
to the constraint solver.
As our experiments show, however, novice programs do not appear to
require more advanced search techniques, likely because they tend to be
simple.


\paragraph{Trace Explosion}
Though the average complexity of our generated traces is low in terms of
jumps, there are some extreme outliers.
We cannot reasonably expect a novice user to explore a trace containing
50+ terms and draw a conclusion about which pieces contributed to the
bug in their program.
Enhancing our visualization to slice out program paths relevant to
specific values~\cite{Perera2012-dy}, would likely help alleviate this
issue, allowing users to highlight a confusing value and ask: ``Where
did this come from?''

\paragraph{Non-Parametric Function Type}
As we discussed in \S~\ref{sec:how-safe} some ill-typed programs
lack a witness in our semantics due to our use of a non-parametric
type $\tfun$ for functions.
These programs cannot ``go wrong'', strictly speaking, but would be very
difficult to \emph{use} in practice.
We also note that many of these programs induce cyclic typing constraints,
causing infinite-type errors, which in our experience can be particularly
difficult to debug (and to explain to novices).
Better support for these programs would be welcome.
For example, we might track how the types of inputs change between
recursive calls.
If we cannot find a traditional witness, we could then produce a trace
expanded to show these particular steps.

\paragraph{Ad-Hoc Polymorphism}

Also discussed in \S~\ref{sec:how-safe}, our approach can only support
ad-hoc polymorphism (\eg\ type-classes in \haskell\ or polymorphic
comparison functions in \ocaml) in limited cases where we have enough
typing information at the call-site to resolve the overloading. This
issue is uncommon in \ocaml\ (we detected it in around 5\% of our
benchmarks), but it would surely be exacerbated by a language like
\haskell, which overloads not only functions but also numeric literals,
as well as strings and lists if one enables the respective language
extensions.
We suspect that either dynamic-symbolic execution or speculative
instantiation of holes would allow us to handle ad-hoc polymorphism, but
defer a proper treatment to future work.

\paragraph{Traversal Bias}

A common problem with typecheckers is that the order in which the
typechecker traverses the abstract syntax tree \emph{biases} it in favor
of blaming expressions that are seen later~\cite{McAdam1998-ub}. This
usually takes the form of a left-to-right bias with respect to the
source code (terms that appear later \emph{textually} are more likely to
be blamed), but in our case the bias is with respect to the execution
trace.

Incorporating our notion of type sources from \S~\ref{sec:locating} into
the visualization, \eg by including those reductions in the initial
visualization, may help alleviate our bias in a similar manner to
McAdam's proposal. Hage and Heeren~\shortcite{Hage2009-em} offer another
solution that allows the compiler author to selectively control the
bias, and thus produce better errors, by prioritizing typing
constraints. Unfortunately, due to \toolname's dynamic nature, providing
this sort of control would likely require selectively changing the order
of evalution; while sound for the pure subset of \ocaml that we address,
this could nonetheless confuse newcomers to the language even more.


\paragraph{Side Effects}

While our implementation does not currently support side effects, we
suspect it would not be difficult to add support.
The search procedure is easily extended to support mutation by
maintaining a store in the evaluation relation, the main complexity
would be in extending the trace visualization to demonstrate mutation.
Incorporating mutation directly into our reduction-based visualization
would be difficult, as mutation would have non-local effects on the
expressions and may be difficult for students to follow.
Instead, we could follow the example of
\textsc{Python Tutor}~\cite{GuoSIGCSE2013}
and provide a separate visualization of the mutable store, with references
visualized as constant pointers to changing objects.



\section{Related Work}
\label{sec:related-work}


\paragraph{Localizing and Repairing Type Errors}
\label{sec:diagnosis-repair}
%
Many groups have explored techniques to improve the error locations
reported by static type checkers.
%
The traditional Damas-Milner type inference algorithm~\cite{Damas1982-uw}
reports the first program location where a type mismatch is discovered
(subject to the traversal strategy~\cite{Lee1998-ys}).
As a result the error can be reported far away from its
source~\cite{McAdam1998-ub} without enough information to guide the
user.
Type-error slicing~\cite{Haack2003-vc,Schilling2011-yf,Rahli2015-tt,Sagonas2013-bf,Gast2004-zd,Neubauer2003-xv}
recognizes this flaw and instead produces a slice of the program
containing \emph{all} program locations that are connected to the type
error.
%
%
Though the program slice must contain the source of the error, it can
suffer from the opposite problem of providing \emph{too much}
information, motivating recent work in ranking the candidate locations.
Zhang~\etal~\shortcite{Zhang2014-lv,Zhang2015-yu} present an algorithm for
identifying the most likely culprit using Bayesian reasoning.
Pavlinovic~\etal~\shortcite{Pavlinovic2014-mr,Pavlinovic2015-kh}
translate the 
localization problem to a MaxSMT optimization problem, using
compiler-provided weights to rank the possible sources.
Loncaric~\etal~\shortcite{Loncaric2016-uk} improve the scalability of
Pavlinovic~\etal by reusing the existing type checker as
a theory solver in the Nelson-Oppen~\shortcite{Nelson1979-td}
style, thus requiring only a MaxSAT solver.

In addition to localizing the error, Lerner~\etal~\shortcite{Lerner2007-dt} attempt to
suggest a fix by replacing expressions (or removing them entirely) with
alternatives based on the surrounding program context.
Chen~\&~Erwig~\shortcite{Chen2014-gd} use a variational type system to allow for the
possibility of changing an expression's type, and search for an
expression whose type can be changed such that type inference would
succeed.
%
%
In contrast to Lerner~\etal, who search for changes at the value-level,
Chen~\&~Erwig search at the type-level and are thus complete due the finite
universe of types used in the program.
%
%

In contrast to these approaches, we do not attempt to localize or fix
the type error. Instead we try to explain it to the user using a
dynamic witness that demonstrates how the program is not just
ill-typed but truly wrong. In addition, allowing users to run their
program (even knowing that it is wrong)
enables experimentation and the
use of debuggers to step through the program and investigate its
evolution.

\paragraph{Improving Error Messages}
The content and quality of the error messages themselves has also been
studied extensively.
Marceau~\etal~\shortcite{Marceau2011-ok,Marceau2011-cy} study the
effectiveness of error messages in novice environments and present
suggestions for improving their quality and consistency.
Hage~\&~Heeren~\shortcite{Hage2006-hc} identify a variety of general
heuristics to improve the quality of type error messages, based on their
teaching experience.
Chargu{\'e}raud~\shortcite{Chargueraud2015-dc} presents a tabular format
for type errors that can provide multiple explanations in a compact form.
Heeren~\etal~\shortcite{Heeren2003-db},
Christiansen~\shortcite{Christiansen2014-qc}, and
Serrano~\&~Hage~\shortcite{Serrano2016-oo}
provide methods for library authors to specialize
type errors with domain-specific knowledge.
The difference with our work is more pronounced here as we do not
attempt to improve the quality of the error message, instead we search
for a witness to the error and explain it with the resulting execution
trace.

\paragraph{Running Ill-Typed Programs}
\label{sec:running-ill-typed}
Vytiniotis~\etal~\shortcite{Vytiniotis2012-gh} extend the \haskell
compiler GHC to support compiling ill-typed programs, but their intent
is rather different from ours. Their goal was to allow programmers to
incrementally test refactorings, which often cause type errors in
distant functions. They replace any expression that fails to type
check with a \emph{runtime} error, but do not check types
at runtime.
Bayne~\etal~\shortcite{Bayne2011-cn} also provide a semantics for running
ill-typed (\java) programs, but in constrast transform the program to
perform nearly all type checking at run-time. The key difference between
Bayne~\etal\ and our work is that we use the dynamic semantics to
automatically search for a witness to the type error, while their focus
is on incremental, programmer-driven testing.

\paragraph{Testing}\label{sec:testing}
\nanomaly is at its heart a test generator, and as such,
builds on a rich line of work.
Our use of holes to represent unknown values is inspired by the work of
Runciman, Naylor, and Lindblad~\cite{Runciman2008-ka,Naylor2007-mi,Lindblad2007-oy},
who use lazy evaluation to drastically reduce the search space for
exhaustive test generation, by grouping together equivalent inputs by
the set of values they force. An exhaustive search is complete (up to
the depth bound), if a witness exists it will be found, but due to the
exponential blowup in the search space the depth bound can be quite
limited without advanced grouping and filtering techniques.
Our search is not exhaustive; instead we use random generation to fill
in holes on demand.
Random test generation~\cite{Claessen2000-lj,Csallner2004-bf,Pacheco2007-at}
is by its nature incomplete, but is able to check larger inputs than
exhaustive testing as a result.

Instead of enumerating values, which may trigger the same path through
the program, one might enumerate paths.
Dynamic-symbolic execution~\cite{Godefroid2005-am,Cadar2008-kg,Tillmann2008-qc}
combines symbolic execution (to track which path a given input triggers)
with concrete execution (to ensure failures are not spurious). The
system collects a path condition during execution, which tracks
symbolically what conditions must be met to trigger the current
path. Upon successfully completing a test run, it negates the path
condition and queries a solver for another set of inputs that satisfy
the negated path condition, \ie inputs that will not trigger the same
path. Thus, it can prune the search space much faster than techniques
based on enumerating values, but is limited by the expressiveness of the
underlying solver.

Our operational semantics is amenable to dynamic-symbolic execution, one
would just need to collect the path condition and replace our
implementation of \gensym by a call to the solver. We chose to use lazy,
random generation instead because it is efficient, does not incur
the overhead of an external solver, and produces high coverage for our
domain of novice programs.

A function's type is a theorem about its behavior.
Thus, \toolname's witnesses can be viewed as \emph{counter-examples},
thereby connecting it to work on using test cases to find
counter-examples prior to starting a proof~\cite{ACL2Testing, Nguyen2015-oo, Seidel15}.


\paragraph{Program Exploration}

Flanagan~\etal~\shortcite{Flanagan96} describe a static debugger for Scheme, which helps
the programmer interactively visualize problematic source-sink flows
corresponding to soft-typing errors. The debugger allows the user to explore
an abstract reduction graph computed from a static value set analysis of
the program. In contrast, \toolname generates witnesses and allows the user
to explore the resulting dynamic execution.

Clements~\etal~\shortcite{Clements2001-qj} present a reduction-based
visualization of program execution similar to \toolname's, though their
interaction model is closer to that of a traditional step-debugger,
limited to taking single step forwards or backwards. In contrast,
\toolname first presents an overview of the whole computation, and then
allows the user to focus in on the interesting reductions.

Perera~\etal~\shortcite{Perera2012-dy} present a tracing semantics
for functional programs that tags values with their provenance, enabling
a form of backwards program slicing from a final value to the sequence
of reductions that produced it. Notably, they allow the user to supply a
\emph{partial value} --- containing holes --- and present a partial slice,
containing only those steps that affected the the partial value.
Perera~\etal\ focus on backward exploration; in contrast, our
visualization supports forward \emph{and} backward exploration, though
our backward steps are more limited.
Specifically, we do not support selecting a value and inserting the
intermediate terms that preceded it while ignoring unrelated computation
steps. 




\section*{Acknowledgments}

We thank Ethan Chan, Matthew Chan and Timothy Nguyen for assisting
with our user study, and we thank the anonymous reviewers and
Matthias Felleisen for their insightful feedback on earlier drafts
of this paper.

{
\bibliographystyle{jfp}
\bibliography{main}
}

\includeTechReport
{

\clearpage
\appendix
\section{Proofs for Section~\ref{sec:searching-witness}}
\label{sec:proofs}

\begin{proof}[Proof of Lemma~\ref{lem:resolve-compat}]
  By induction on $\trace$.
  In the base case $\trace = \triple{\eapp{f}{\vhole{\thole}}}{\emptysu}{\emptysu}$
  and $\thole$ is trivially a refinement of $\vhole{\thole}$.
  In the inductive case, consider the single-step extension of $\trace$,
  $\trace' = \trace,\triple{e'}{\vsu'}{\tsu'}$.
  We show by case analysis on the evaluation rules that if
  $\vsub{\resolve{\thole}{\tsu}}{\resolve{\ehole}{\vsu}}$, then
  $\vsub{\resolve{\thole}{\tsu'}}{\resolve{\ehole}{\vsu'}}$.

  We can immediately discharge all of the \rulename{E-*-Bad} rules
  (except for \renodebadone) as the calls to \forcesym\ return \stuck.
  An examination of \forcesym shows that if \forcesym\ returns \stuck\
  then \vsu\ and \tsu\ are unchanged.
  \begin{description}
  \item[Case \replusgood:]
    We \forcesym\ $v_1$ and $v_2$ to $\tint$, so we must consider
    the $\force{\vhole{\thole}}{t}{\vsu}{\tsu}$ and
    $\force{n}{\tint}{\vsu}{\tsu}$ cases.
    The $\force{n}{\tint}{\vsu}{\tsu}$ case is trivial as it does
    not change \vsu\ or \tsu.
    In the $\force{\vhole{\thole}}{t}{\vsu}{\tsu}$ we will either
    find that $\ehole \in \vsu$ or we will generate a fresh \tint
    and extend \vsu.
    Note that when we extend \vsu\ we also extend \tsu\ due to the
    call to \unifysym, thus in the $\vhole \in \vsu$ we cannot
    actually refine either $\ehole$ or $\thole$ and thus the
    refinement is preserved.
    When we extend \vsu\ with a binding for \ehole, the call to
    \unifysym ensures that we add a compatible binding for
    \thole if one was not already in \tsu, thus the refinement
    relation must continue to hold.
  \item[Case \rulename{E-If-Good\{1,2\}}:]
    Similar to \replusgood.
  \item[Case \reappgood:]
    Similar to \replusgood.
  \item[Case \releafgood:]
    This step cannot change \vsu\ or \tsu\ thus the refinement
    relation continues to hold trivially.
  \item[Case \renodegood:]
    We \forcesym\ $v_2$ and $v_3$ to $\ttree{t}$, so we must consider
    three cases of \forcesym.
    \begin{description}
    \item[\force{\vhole{\thole}}{t}{\vsu}{\tsu}:]
      Similar to \replusgood.
    \item[\force{\vleaf{t_1}}{\ttree{t_2}}{\vsu}{\tsu}:]
      This case may extend \tsu\ but not \vsu, so the refinement
      continues to hold trivially.
    \item[\force{\vnode{t_1}{v_1}{v_2}{v_3}}{\ttree{t_2}}{\vsu}{\tsu}:]
      Same as \vleaf{t_1}.
    \end{description}
  \item[Case \rulename{E-Case-Good\{1,2\}}:]
    Similar to \replusgood.
  \item[Case \repcasegood:]
    Similar to \replusgood.
  \end{description}
\end{proof}

\begin{proof}[Proof of Lemma~\ref{lem:k-stuck}]
  We can construct $v$ from $\trace$ as follows.
  Let
  $$
  \trace_i = \triple{\eapp{f}{\vhole{\thole}}}{\emptysu}{\emptysu},
             \ldots,
             \triple{e_{i-1}}{\vsu_{i-1}}{\tsu_{i-1}},
             \triple{e_{i}}{\vsu_{i}}{\tsu_{i}}
  $$
  be the shortest prefix of $\trace$ such that
  $\tincompat{\ptype{\trace_i}{f}}{t}$.
  %
  We will show that $\ptype{\trace_{i-1}}{f}$ 
  must contain some other hole $\thole'$ that is
  instantiated at step $i$.
  Furthermore, $\thole'$ is instantiated in such a way that
  $\tincompat{\ptype{\trace_i}{f}}{t}$.
  Finally, we will show that if we had instantiated $\thole'$ such that
  $\tcompat{\ptype{\trace_i}{f}}{t}$,
  the current step would have gotten $\stuck$.

  %
  Since $\tsu_{i-1}$ and $\tsu_{i}$ differ only in $\thole'$ but the resolved
  types differ, we have
  $\thole' \in \ptype{\trace_{i-1}}{f}$
  and
  $\ptype{\trace_i}{f} = \ptype{\trace_{i-1}}{f}\sub{\thole'}{t'}$.
  %
  Let $s$ be a
  concrete type such that $\ptype{\trace_{i-1}}{f}\sub{\thole'}{s} = t$.
  We show by case analysis on the evaluation rules that
  $$\step{e_{i-1}}{\vsu_{i-1}}{\tsu_{i-1}[\thole' \mapsto s]}{\stuck}{\vsu}{\tsu}$$
    \begin{description}
    \item[Case \replusgood:]
      Here we \forcesym $v_1$ and $v_2$ to $\tint$, so the first case of
      \forcesym must apply ($\force{n}{\tint}{\vsu}{\tsu}$ cannot apply
      as it does not change $\tsu$).
      In particular, since we extended $\tsu_{i-1}$ with
      $[\thole' \mapsto t']$ we know that $\thole' = \thole$ and
      $t' = \tint$.
      Let $s$ be any concrete type that is incompatible with $\tint$
      and $\tsu_s = \tsu_{i-1}[\thole \mapsto s]$,
      $\force{\vhole{\thole}}{\tint}{\vsu_{i-1}}{\tsu_s]} = \triple{\stuck}{\vsu_{i-1}}{\tsu_s}$.
    \item[Case \rulename{E-Plus-Bad\{1,2\}}:]
      These cases cannot apply as \forcesym does not update \tsu\ when
      it returns \stuck.
    \item[Case \rulename{E-If-Good\{1,2\}}:]
      Similar to \replusgood.
    \item[Case \reifbad:]
      This case cannot apply as \forcesym does not update \tsu\ when
      it returns \stuck.
    \item[Case \reappgood:]
      Similar to \replusgood.
    \item[Case \reappbad:]
      This case cannot apply as \forcesym does not update \tsu\ when
      it returns \stuck.
    \item[Case \releafgood:]
      This case cannot apply as it does not update \tsu.
    \item[Case \renodegood:]
      Here we \forcesym $v_2$ and $v_3$ to $\ttree{t}$,
      so we must consider three cases of \forcesym.
      \begin{description}
      \item[\force{\vhole{\thole}}{t}{\vsu}{\tsu}:]
        Similar to \replusgood.
      \item[\force{\vleaf{t_1}}{\ttree{t_2}}{\vsu}{\tsu}:]
        For this case to extend \tsu\ with $[\thole' \mapsto t']$,
        either $t_1$ or $t_2$ must contain $\thole'$.
        Let $s$ be any concrete type that is incompatible with $t'$
        and $\tsu_s = \tsu_{i-1}[\thole \mapsto s]$,
        $\force{\vhole{\thole}}{\tint}{\vsu_{i-1}}{\tsu_s]} = \triple{\stuck}{\vsu_{i-1}}{\tsu_s}$.
      \item[\force{\vnode{t_1}{v_1}{v_2}{v_3}}{\ttree{t_2}}{\vsu}{\tsu}:]
        Same as \vleaf{t_1}.
      \end{description}
    \item[Case \renodebadone:]
      This case cannot apply as \forcesym does not update \tsu\ whe
      it returns \stuck.
    \item[Case \renodebadtwo:]
      Similar to \renodegood.
    \item[Case \rulename{E-Case-Good\{1,2\}}:]
      Here we \forcesym $v$ to $\ttree{\thole}$,
      so we must consider three cases of \forcesym.
      \begin{description}
      \item[\force{\vhole{\thole}}{t}{\vsu}{\tsu}:]
        Similar to \replusgood.
      \item[\force{\vleaf{t_1}}{\ttree{t_2}}{\vsu}{\tsu}:]
        This case cannot extend \tsu\ with $[\thole' \mapsto t']$ as we
        use a fresh $\thole$, which cannot be referenced by
        $\ptype{\trace_{i-1}}{f}$, in the call to \forcesym, and thus it
        cannot apply.
      \item[\force{\vnode{t_1}{v_1}{v_2}{v_3}}{\ttree{t_2}}{\vsu}{\tsu}:]
        Same as \vleaf{t_1}.
      \end{description}
    \item[Case \recasebad:]
      This case cannot apply as \forcesym does not update \tsu\ whe
      it returns \stuck.
    \item[Case \repcasegood]
      Here we \forcesym $v$ to $\tprod{\thole_1}{\thole_2}$,
      so we must consider two cases of \forcesym.
      \begin{description}
      \item[\force{\vhole{\thole}}{t}{\vsu}{\tsu}:]
        Similar to \replusgood.
      \item[\force{\epair{v_1}{v_2}}{\tprod{t_1}{t_2}}{\vsu}{\tsu}:]
        This case cannot extend \tsu\ with $[\thole' \mapsto t']$ as we
        use a fresh $\thole_1$ and $\thole_2$, which cannot be
        referenced by $\ptype{\trace_{i-1}}{f}$, in the call to
        \forcesym, and thus it cannot apply.
      \end{description}
    \item[Case \repcasebad:]
      This case cannot apply as \forcesym does not update \tsu\ whe
      it returns \stuck.
    \end{description}
  Finally, by Lemma~\ref{lem:resolve-compat} we know that
  $\vsub{\ptype{\trace_{i-1}}{f}}{\resolve{\ehole}{\vsu_{i-1}}}$
  and thus $\thole' \in \resolve{\vhole{\thole}}{\vsu_{i-1}}$.
  %
  Let $u = \gen{s}{\tsu}$
  and $v = \resolve{\ehole}{\vsu_{i-1}}\sub{\ehole'[\thole']}{u}\sub{\thole'}{s}$,
  $\steps{\eapp{f}{v}}{\emptysu}{\emptysu}{\stuck}{\vsu}{\tsu}$ in $i$ steps.
\end{proof}
\clearpage


\section{User Study}
\label{sec:user-study-exams}
\subsection{Version A}
\fbox{\includegraphics[width=\textwidth,page=4]{user-study.pdf}}
\newpage
\fbox{\includegraphics[width=\linewidth,page=5]{user-study.pdf}}
\newpage
\fbox{\includegraphics[width=\linewidth,page=6]{user-study.pdf}}
\newpage
\subsection{Version B}
\fbox{\includegraphics[width=\linewidth,page=1]{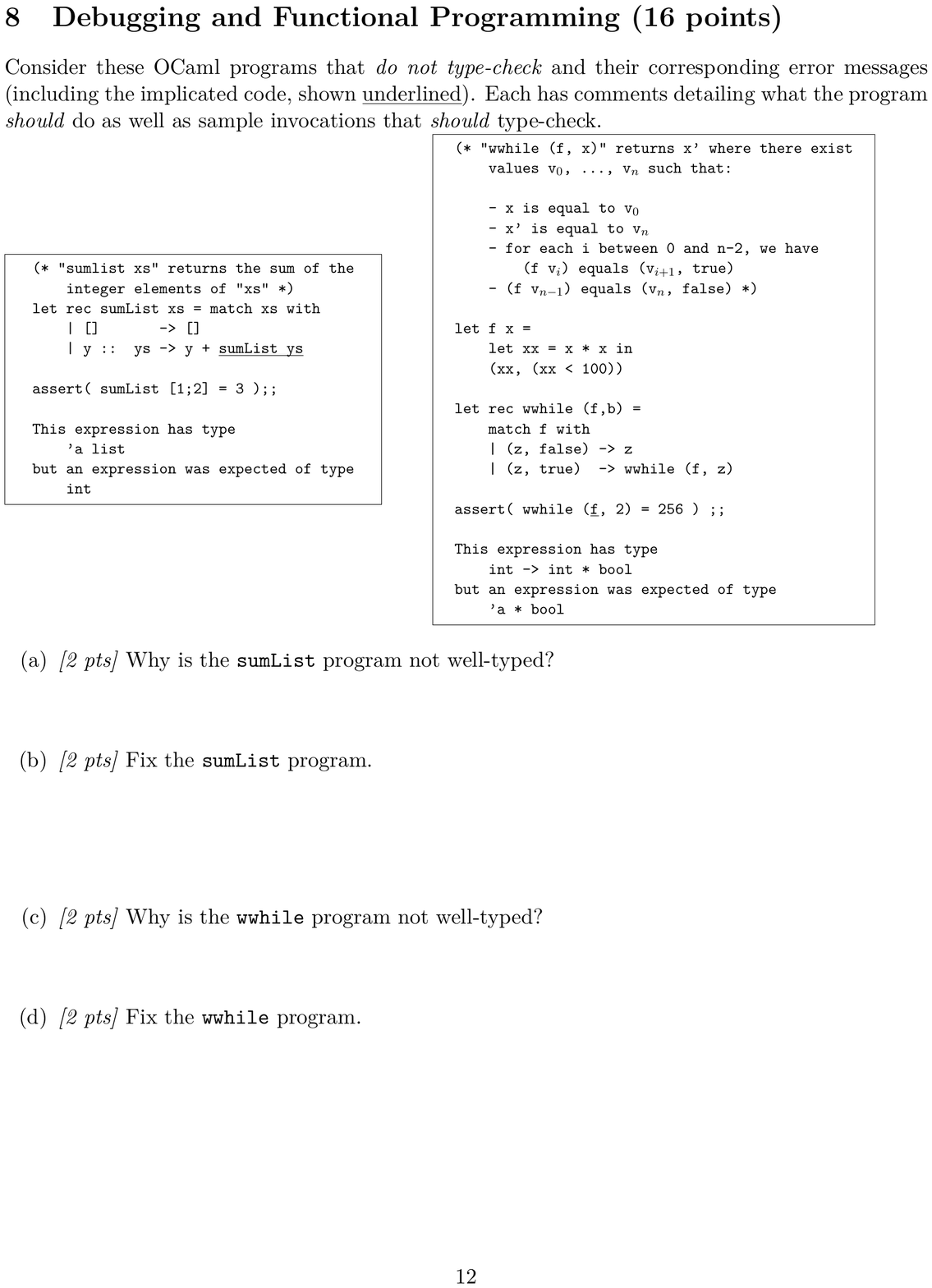}}
\newpage
\fbox{\includegraphics[width=\linewidth,page=2]{user-study.pdf}}
\newpage
\fbox{\includegraphics[width=\linewidth,page=3]{user-study.pdf}}
}

\end{document}